\documentclass[11pt,a4paper]{article}

\pdfoutput=1 

\usepackage{xcolor}
\newcommand{\TODO}[1]{\textcolor{red}{\textbf{[TODO: #1]}}}

\usepackage{booktabs}
\usepackage{tikz}
\usepackage{subcaption}
\usetikzlibrary{fit,backgrounds}
\usetikzlibrary{arrows.meta}
\usepackage{float}

\usepackage[margin=1in]{geometry}
\usepackage{amsmath,amssymb,amsthm}
\usepackage{hyperref}
\usepackage{cleveref}

\newtheorem{theorem}{Theorem}[section]
\newtheorem{proposition}[theorem]{Proposition}
\newtheorem{corollary}[theorem]{Corollary}
\newtheorem{lemma}[theorem]{Lemma}
\newtheorem{definition}[theorem]{Definition}

\newcommand{\cI}{\ensuremath{\mathcal{I}}}
\providecommand{\OPT}{\operatorname{OPT}}

\usepackage{comment}

\title{Stochastic Minimum Spanning Tree Identification with Binary Weights} 

\author{
  Haya Diwan\thanks{Department of Computer Science and Engineering, New York University, USA. \\Email: \texttt{hd2371@nyu.edu}.}
  \and
  Lisa Hellerstein\thanks{Department of Computer Science and Engineering, New York University, USA. \\Email:\texttt{lisa.hellerstein@nyu.edu}.}
  \and
  Sanjeev Khanna \thanks{Department of Computer Science, Courant Institute, New York University, USA. Supported in part by NSF award CCF-2625203 and AFOSR award FA9550-25-1-0107. Email: \texttt{sanjeev.khanna@nyu.edu}.}
  \and 
  Junkai Song\thanks{Department of Computer Science, Courant Institute, New York University, USA. Supported in part by NSF award CCF-2625203. Email: \texttt{junkaisong@nyu.edu}.}
}

\begin{document}

\maketitle

\begin{abstract}
We introduce the Stochastic MST Identification Problem, with binary weights.  In this problem, we need to identify a minimum spanning tree (MST) of a graph, each of whose edge weights is either 0 or 1.  The weight of each edge $e$ is a random variable which is 1 with probability $p_e$ and 0 with probability $1-p_e$.  The probability $p_e$ is given for each edge $e$, but the realized weight of an edge can only be determined by querying the edge. The problem is then to determine a strategy for querying the edges of the graph that minimizes the expected number of queries that need to be performed until we can identify a subset of edges that form an MST of the graph. We address both adaptive and non-adaptive versions of this problem. For the non-adaptive version of our Stochastic MST Identification problem, we give an $O(\log E)$-approximation algorithm via a reduction to submodular ranking. For the adaptive version, we show that solving the problem is trivial when all $p_e$ are equal to 1/2, because of a relation to the Tutte polynomial.  However, computing the expected cost of an optimal querying strategy is \#P-hard, even when all $p_e$ equal $\frac{1}{2}$.
For arbitrary $p_e$ values, we present a constant-factor approximation algorithm for simple planar graphs. We investigate the question of whether 
the {\em submodular goal value} approach of Deshpande et al.\ could be used to achieve a $\log$-factor approximation algorithm that would work for arbitrary $p_e$ values and general graphs.  In this context, we address the structural question of determining optimal submodular goal values for MST identification in graphs.  
\end{abstract}

\newpage

\section{Introduction}\label{sec:intro}

Consider the general problem of identifying a minimum-weight spanning tree (MST) in an {\em uncertainty graph}.  In such a graph, the weights of the edges are not initially known.  For each edge $e$, we are guaranteed that its weight $w_e$ is a member of a given {\em uncertainty set} $A_e$.  But we can only determine the actual value of $w_e$ by asking a query.  Querying must continue until enough information is obtained to identify the edges of a minimum spanning tree (MST) of the graph.  However, it is not necessary to discover the weight of the MST.

The problem of identifying an MST with queries under uncertainty was introduced by Erlebach and Hoffmann~\cite{erlebach08steiner_uncertainty} in an online adversarial setting, where querying is done adaptively and the goal is to minimize the worst-case competitive ratio against the optimal offline {\em MST Verification} cost. Erlebach and Hoffmann ~\cite{erlebach_verification} also introduced the study of the offline MST Verification problem. It takes as input the uncertainty sets $A_e$ and the exact edge weights $w_e$, and the goal is to identify a min-cost set of queries verifying an MST. Querying all edges in this set certifies an MST of the graph regardless of the values of the unqueried edges. 

In this paper, we introduce the study of the MST identification problem in the stochastic setting, where the weight $w_e$ of each edge is drawn independently from a given probability distribution on $A_e$.  The goal in this setting is to develop a query strategy that minimizes the expected cost of the performed queries.

We focus on the fundamental binary weight case, where each $w_e$ is either 0 or 1 (equivalently, $A_e=\{0,1\}$).  For each edge $e$, we have $p_e = P[w_e = 1]$.
We set each query to have unit cost, so the goal is to minimize the expected {\em number} of queries.  We call this the Stochastic MST Identification problem with binary weights, or the MST-ID problem, for short. 

We note that in the online adversarial setting, the binary weight case is  not interesting, because it is not possible to achieve sublinear competitive ratio~\cite{erlebach08steiner_uncertainty}.
However, as shown recently by Diwan et al. \cite{diwan2025} in the verification setting the binary case is interesting and non-trivial, because different MSTs of the same uncertainty graph can have different optimal verification costs~\cite{diwan2025}. 

We consider both adaptive and non-adaptive versions of the MST-ID problem.  In the non-adaptive version, we only allow non-adaptive query strategies.  These are represented by a permutation of the edges; queries must be performed in the order specified by the permutation until enough information is obtained to identify an MST of the graph.    
In the adaptive version, 
we allow adaptive query strategies, which correspond to decision trees.

The adaptivity gap for the MST-ID problem is large; it can easily be shown using arguments from \cite{Hellerstein_Kletenik_Liu_Witter_2022} that there are instances for which the expected number of queries of the best non-adaptive strategy is $\Omega(\sqrt{n})$ times larger than that of the best adaptive strategy. 

Our work leaves open the question of whether the adaptive and non-adaptive versions of the MST-ID problems are NP-hard.  It is unclear whether resolving this question will require new techniques. We note that there are a number of other stochastic probing problems whose NP-hardness remains open, most notably the SBFE problem for read-once Boolean formulas, which has been studied since the 1970's~(cf.~\cite{unluyurt2004sequential, unluyurt2025sequential}).

We do show that for the adaptive version of MST-ID, computing the expected cost of the optimal  strategy is \#P-hard.  This holds even in the special case where all $p_e$ are equal to 1/2.  In contrast, as discussed below, the adaptive MST-ID problem itself has an easy polynomial-time algorithm in this special case. Our main algorithmic results are as follows.

\vspace{2mm}
{\bf \noindent Non-adaptive MST-ID}~~
We give a polynomial-time $O(\log E)$-approximation algorithm for the non-adaptive MST-ID problem.  This algorithm is based on an approach recently used in \cite{brody2026extensionsunanimousvoteproblem}, which involves reducing a non-adaptive stochastic probing problem to the problem of Submodular Ranking. To apply this approach to MST-ID, we construct, for each realization $\alpha \in \{0,1\}^{E}$ of the edge weights, a corresponding monotone submodular {\em goal function} $g_{\alpha}:2^{E} \rightarrow \{0,1,2,\ldots,Q_{\alpha}\}$ such that $Q_{\alpha}$ is polynomial in $|E|$ for all $\alpha$. 

\vspace{2mm}
{\bf \noindent Adaptive MST-ID}~~ 
For the adaptive version of the MST-ID problem, we begin by showing that in the special case where all $p_e$ are equal to $1/2$, the problem is very easy. 
It follows from a relation to the Tutte polynomial that in this special case, it is optimal to query the edges in any order, provided that ``useless'' queries are skipped.  

However, for arbitrary $p_e$, the adaptive version of the MST-ID problem appears to be much more challenging.  
As part of addressing it, we consider the corresponding adaptive {\em MST Weight Determination problem (MST-WD)}, where the goal is not just to identify the MST, but also to determine its weight.  
We note that because weights are binary, the problem of computing the weight of an MST of the graph is equivalent to the problem of computing the number of connected components of the subgraph consisting of the weight 0 edges.

We show that the adaptive MST-WD problem has an easy exact algorithm.  Further, we show that on simple planar graphs, that algorithm gives a constant-factor approximation for the adaptive MST-ID problem. 

If all $p_e$ are greater than 1/2 (symmetrically, all are less than 1/2), then performing queries in increasing (decreasing)  $p_e$ order and skipping ``useless'' queries is a 2-approximation algorithm. 

We leave as an open question whether there is a polynomial-time approximation algorithm that achieves a non-trivial approximation factor on general graphs with arbitrary probabilities.  

In addition to the above algorithmic results, we also consider related questions about submodular goal value.
As discussed above, 
the approach we use to achieve our polynomial-time $O(\log E)$ approximation algorithm for the non-adaptive MST-ID problem is based on a reduction to Submodular Ranking.  This approach is only suitable for non-adaptive problems. An analogous approach, used by Deshpande et al.\ for Stochastic Boolean Function Evaluation (SBFE) problems \cite{SBFE}, is suitable for adaptive problems.  It is based on a reduction to Stochastic Submodular Cover.  

Using this analogous approach to achieve a logarithmic approximation algorithm for adaptive MST-ID requires the construction of a state-dependent monotone submodular goal function $f:\{0,1,*\}^E \rightarrow \{0,1,2,\ldots,Q\}$ where $Q$ is polynomial in $|E|$.  In contrast to the functions $g_{\alpha}$ constructed for the non-adaptive problem, which can be different for each realization $\alpha$, 
this is a single function whose domain is the set of all partial realizations of the weights of the edges in $E$.  
This gives rise to the structural question: does such a single function exist?  
Or phrased another way:  is there a monotone submodular goal function for this problem that has polynomial goal value?  

We do not resolve this question, but we give related preliminary results and observations.  Here we note that
the same question has previously been considered in the context of SBFE problems, where the goal is evaluation of a Boolean function~\cite{bach2018submodular}.  In that context, the answer was shown to be ``no'' for certain Boolean functions. The proof was implicitly based on showing an exponential lower bound on goal value for one of the $g_{\alpha}$ for the non-adaptive problem.  However, as discussed above, for MST identification we have shown that for every $\alpha$, there is a $g_{\alpha}$ with  polynomial goal value.  
Thus proving exponential lower bounds on goal value for MST identification would require new techniques.  

\subsection{Related work}

For the offline MST verification problem, Erlebach and Hoffmann~\cite{erlebach_verification} gave a polynomial time algorithm for computing a minimum cost query set when every uncertainty set is an open interval or a single point. Recently,  Diwan et al. \cite{diwan2025}
showed that an optimal query set can be computed exactly in polynomial time for arbitrary uncertainty sets and query costs, by proving this holds for the more general problem of identifying a minimum weight basis of a matroid. In the online adversarial setting, constant competitive ratios have been established for both MST and matroid basis identification. Erlebach et al. \cite{erlebach2016query} proved a best possible ratio of $2$ for deterministic algorithms. Megow et al. proved a $1.707$ ratio for a randomized algorithm \cite{megow2017randomization}.

Stochastic MST Identification is an optimization problem in which the task is to determine an order in which to perform queries, so as to achieve some goal (identifying an MST of the input graph) while incurring minimal expected cost.
The SBFE problem, mentioned above, is of the same general type. There the goal is to determine the value of a given Boolean function on
a random assignment of Boolean values to its input variables; the value assigned to an input variable can only be determined by performing a test (asking a query). SBFE problems have been studied for many years, see the survey papers of Ünlüyurt~\cite{unluyurt2004sequential, unluyurt2025sequential}.  Stochastic Function Evaluation (SFE) problems have also been studied for non-Boolean functions (see e.g.,\cite{ghuge2025nonadaptive, grammel2022algorithms, plank2024simple,BallotCounting}).  
As discussed above, we use techniques from the SBFE literature in our work.  
However, the MST-ID problem poses novel challenges because of its graph-theoretic nature.

Our $O(\log E)$-approximation algorithm for non-adaptive MST-ID is based on the approach introduced by Brody et al.\cite{brody2026extensionsunanimousvoteproblem}, who used it to solve a dice-rolling problem they called Finite Coupon Collection.  The greedy algorithm for Submodular Ranking used in that approach was introduced by Azar and Gamzu~\cite{AzarGamzu}. Im et al. \cite{Im2016} subsequently provided an alternative analysis of the algorithm. Their work also shows that a natural generalization of the rule used in the greedy algorithm achieves the same approximation factor when different queries have different costs.

The MST-ID problem is an example of an \emph{explorable uncertainty} problem, where input parameters are initially unknown but can be revealed through queries to compute an optimal solution.   Merino and Soto \cite{MerinoS19} studied the problem of finding a minimum weight basis of an uncertainty matroid in a setting where 
there is no prior knowledge regarding the actual element weights. They require that their query set act as a certificate across every realization such that $w_e \in A_e$ for all elements $e$. 
In contrast, we consider a stochastic setting  where query strategies must minimize expected query cost over a known distribution. 

Explorable uncertainty in stochastic settings has also been studied for graph and hypergraph orientation \cite{BampisDEdLMS21}, sorting and median selection \cite{ChaplickHLT20}, and covering based set selection \cite{MegowS23}. These prior works assume real valued weights drawn from continuous distributions. There is also related work that studies explorable uncertainty with untrusted predictions for MST and sorting/orientation \cite{ErlebachLMS23, ErlebachLMS22}.

MST problems have also been studied in stochastic settings that are fundamentally different from ours, such as the two-stage stochastic optimization with recourse model \cite{dhamdhere2005pay, dhamdhere2005two, flaxman2006random}, and related two-stage combinatorial frameworks \cite{gupta2004boosted, gupta2004edge, immorlica2004costs, ravi2004hedging, shmoys2004stochastic}. In these models, decisions center on purchasing edges before and after uncertainty resolves in order to minimize the total edge usage cost. 

\section{Preliminaries} \label{sec: prelim}

\subsection{Basic definitions and notation}

{\bf \noindent Miscellaneous.} We use ``increasing'' (resp. ``decreasing'') to mean ``non-decreasing'' (resp. ``non-increasing'').  In writing expressions involving $|E|$ and $|V|$, we sometimes use $E$ and $V$ instead. Additionally, we note that full proofs for all results stated without a proof in the main text, along with supplementary figures and experiments, are deferred to the Appendix.

\vspace{2mm}
{\bf \noindent Graphs, minors, and certification.} 
We do not assume that graphs are simple; they may contain loops and multi-edges (also called \emph{parallel edges}), which are distinct edges sharing the same pair of endpoints. {\em Contracting} an edge $e=(v,v')$ in a graph $G=(V,E)$ corresponds to merging the endpoints $v,v'$ of $e$ into a single vertex $\hat{v}$ and deleting $e$. While technically the endpoints of edges adjacent to $v$ and $v'$ change when those vertices are merged, we usually treat the edges of the contracted graph as being equal to the edges in $G$ from which they were derived. 

For an edge $e \in E$, $G \setminus e$ denotes the deletion of $e$ from $G$, and $G / e$ denotes the contraction of $e$ in $G$. 
The edge set of both these induced graphs is $E \setminus \{e\}$. 

A graph created from $G$ by a sequence of edge contractions, edge deletions, and deletions of isolated vertices is a {\em minor} of $G$. A \emph{bridge} of $G$ is an edge whose removal increases the number of connected components of $G$. A \emph{loop} is an edge whose endpoints are identical. 

Let $T$ be some spanning tree of a connected graph $G$. Tree $T$ must contain all bridge edges of $G$. If $e \in E \setminus T$, then $e$ forms a cycle with edges of $T$ called the \emph{fundamental cycle} of $e$ with $T$.

A (full) \emph{realization} of edge weights is a vector $\alpha = \{0, 1\}^E$, where $\alpha_e \in \{0, 1\}$ represents the realized weight $w_e = \alpha_e$ of some edge $e$. For $\alpha \in \{0,1\}^E$ and a spanning tree $T$ of $G$, we use $\alpha(T) = \sum_{e \in T} \alpha_e$ to denote the weight of $T$ under realization $\alpha$. We say $T$ is an MST of $G$ for weights $\alpha$ if $\alpha(T) \le \alpha(T')$ for every spanning tree $T'$ of $G$. 
A \emph{partial realization} is a vector $\beta \in \{0, 1, *\}^E$.  Partial realizations are used to represent information about edge weights that has been obtained through queries: $\beta_e \in \{0, 1\}$ indicates that edge $e$ has been queried and found to have weight $\beta_e$ and $\beta_e = *$ indicates that 
$e$ has not been queried and its weight is unknown.  For partial realizations $\beta,\gamma \in \{0,1,*\}^E$ we call $\beta$ an extension of $\gamma$, or say that $\beta$ contains $\gamma$ (written $\gamma \preceq \beta$) if $\beta_e=\gamma_e$ for all $e \in E$ with $\gamma_e \neq *$.  

Given a partial realization $\beta \in \{0,1,*\}^E$, $G_\beta$ is the minor of $G$ obtained by contracting all edges with $\beta_e = 0$ and deleting all edges with $\beta_e = 1$. 
For a realization $\alpha \in \{0,1\}^E$ and an edge subset $S \subseteq E$, we define $\alpha^S \in \{0,1,*\}^E$ to be the partial realization where $\alpha^S_e = \alpha_e$ for $e \in S$ and $\alpha^S_e = *$ for $e \notin S$. 

Let $T$ be a spanning tree of graph $G$. A partial realization $\beta \in \{0,1,*\}^E$ \emph{certifies} $T$ as an MST of $G$ if for every (full) realization $\alpha$ such that $\beta \preceq \alpha$,
$T$ is an MST of $G$  for weights $\alpha$.  
 We say a partial realization $\beta$ is an \emph{MST-ID certificate for  graph $G$} if it certifies at least one MST $T$ of $G$. 

\vspace{2mm}
{\bf \noindent The Stochastic MST-ID problem.} 
The Stochastic MST-ID problem was described in Section~\ref{sec:intro}.  We give some additional information about the problem here.  We represent an instance of the problem as a pair $\mathcal{I} = (G, p)$, where $G = (V, E)$ is the connected graph and $p = (p_e)_{e\in E}\in (0,1)^E$ is the vector of edge probabilities $p_e = P[w_e=1]$.  
Given an oracle that answers the edge queries, an (exact) algorithm for solving the MST-ID problem must ask queries according to a strategy with minimum expected cost (minimum expected number of queries), and end by outputting an identified MST.
An identified MST is a tree $T$ that is certified as an MST of $G$ by $\beta$, the partial realization representing the answers received to the performed queries.  

The Stochastic MST-ID problem can be contrasted with the corresponding offline MST verification problem.  That problem takes as input the graph $G$ and a realization $\alpha \in \{0,1\}^E$ of its edge weights. The problem is to find a minimum-cost MST-ID certificate $\beta$ for $G$ such that $\beta \preceq \alpha$. The cost of $\beta$ is the number of queries performed before obtaining such a certificate, which is $|\{e \in E \mid \beta_e \neq *\}|$.

\vspace{2mm}
{\bf \noindent Monotone submodular functions.}
Let $I$ be a finite ground set of items. We call $f: 2^I \rightarrow \mathbb{R}_{\geq 0}$ a \emph{set function}. For $S \subseteq I$, we consider $f(S)$ to be the value of the set $S$, with respect to $f$. A set function $f$ is \emph{monotone} if $f(A) \leq f(B)$ for all $A \subseteq B \subseteq I$. It is \emph{submodular} if $f(A \cup \{i\}) - f(A) \geq f(B \cup \{i\}) - f(B)$ for all $A \subseteq B \subseteq I$ and $i \in I \setminus B$. 

\subsection{Important properties of graphs and MSTs}

We will rely on the following proposition, which was fundamental to previous work on MST identification in online and verification settings.  

\begin{proposition} 
\label{prop:fundcyc}
Let $G=(V,E)$ be a connected graph and let $T$ be a spanning tree of $G$.  
For each $e \in E\setminus T$, let $C_e$ denote the fundamental cycle that $e$ forms with $T$.
Then $\beta \in \{0,1,*\}^E$ is a certificate that $T$ is an MST of $G$ iff the following holds for every $e \in E \setminus T$:  
if $\beta_e = *$ or $\beta_e = 0$, then $\beta_{e'}=0$ for all $e' \in C_e \setminus \{e\}$ .
\end{proposition}

We use the above proposition to prove the following lemma, which characterizes when a partial realization $\beta$ has enough information to identify an MST of $G$.

\begin{lemma} \label{prop: beta certifies if no loop or bridge}
Let $G=(V,E)$ be a connected graph and let $\beta \in \{0,1,*\}^E$.  Then $\beta$ is an MST-ID certificate on graph $G$ iff every edge of $G_{\beta}$ is either a loop or a bridge.   
\end{lemma}

\begin{proof} \leavevmode

    {\bf \noindent $\Rightarrow$} Suppose
    for contradiction that $\beta$ certifies some $T$ as an MST of $G$, but
    $G_\beta$ has an edge $e^*$ that is not a loop or a bridge. 
    Then $e^*$ lies on a non-loop cycle in $G_\beta$.   Recall that $G_{\beta}$ is formed from $G$ by contracting edges $e$ with $\beta_{e}=0$ and deleting edges $e$ with $\beta_{e}=1$.  
    Thus $e^*$ also lies on a cycle $C$ of $G$ in which each edge $e$ satisfies $\beta_e=*$ or $\beta_e=0$.  

    We claim that $e^* \in T$.  Suppose not. Then by Proposition~\ref{prop:fundcyc}, $\beta_e = 0$ for all other edges in the fundamental cycle of $e^*$ with $T$. So those edges were contracted in forming $G_{\beta}$, making $e^*$ a loop of $G_{\beta}$.  Contradiction.  Thus $e^*$ is an edge of $T$.
    
    Let $(V',V'')$ be the cut of $G$ defined by removing $e^*$ from $T$ and letting $V'$ and $V''$ be the vertices in the resulting connected components.  Since $C$ is a cycle containing $e^*$, there must be at least one other edge $e'$ of $C$ crossing this cut.  Since $\beta$ certifies $T$ as an MST, it can be shown that  $\beta_{e'}=1$, which contradicts that each edge $e$ in $C \setminus e^*$ satisfies $\beta_e=*$ or $\beta_e=0$.  
    
\smallskip
    {\bf \noindent $\Leftarrow$} Suppose that $G_{\beta}$ has only loops and bridges.  Construct a spanning tree $T$ of $G$ by
    running Kruskal's algorithm on $G$, using the 0 or 1 weights given by $\beta$ for all $e$ such that $\beta_e \neq *$, and using weight $\frac{1}{2}$ for edges with $\beta_e=*$.  We now prove that $\beta$ certifies the resulting tree $T$ as an MST of $G$, using Proposition~\ref{prop:fundcyc}.  Consider $e \in E \setminus T$.  
    If $\beta_e=0$, then because $e$ was not placed in $T$ during execution of Kruskal's algorithm, it formed a cycle with edges $e'$ in $T$ such that $\beta_{e'}=0$.  That cycle is the fundamental cycle of $e$ with $T$, so the condition of Proposition~\ref{prop:fundcyc} is satisfied.  Suppose $\beta_e=*$.  
    Then because it was not placed in $T$ during execution of Kruskal's algorithm, it formed a cycle with edges $e'$ in $T$ such that $\beta_{e'}=0$ or $\beta_{e'}=*$.  Thus $e$ lies on a cycle in $G_{\beta}$, so $e$ is a loop in $G_{\beta}$. Then there is a path in $G$ joining the endpoints of $e$ whose edges $e'$ all have $\beta_{e'}=0$.  It follows that when $e$ was processed by Kruskal's algorithm, 
    there was already a (possibly different) path of edges $e'$ in $T$ satisfying $\beta_{e'}=0$, forming the fundamental cycle of $e$ with $T$.   
    So the condition of Proposition~\ref{prop:fundcyc} is satisfied. 
    \end{proof}
    
The above lemma and its proof have implications for algorithms solving the MST-ID problem.  
As such an algorithm is making queries, it can 
easily determine whether it has already obtained enough information to identify an MST of $G$, simply by checking whether $G_{\beta}$ has only loops and bridges, for partial realization $\beta$ representing the answers to the queries performed so far.  
Further, once this point has been reached, it can construct an appropriate MST $T$ using Kruskal's algorithm.   

\subsection{Matroids} \label{sec: matroid generalization prelims}

Some of the results in this paper use matroid theory. Here we present essential relevant definitions. For a more information on matroid theory see Oxley \cite{oxley}.

\begin{definition}[Matroid Basics] \label{def: matroid basics}
    A matroid $M = (E, \mathcal{I})$ consists of a finite ground set $E$ and a family $\mathcal{I} \subseteq 2^E$ of subsets, called \emph{independent sets}, satisfying: (i) ~$\emptyset \in \mathcal{I}$, (ii)~if $I \in \mathcal{I}$ and $I' \subseteq I$, then $I' \in \mathcal{I}$ (downward closure), and (iii)~if $I, J \in \mathcal{I}$ with $|I| < |J|$, there exists some $e \in J \setminus I$ such that $I \cup \{e\} \in \mathcal{I}$.

    A maximal independent set is a \textit{basis}. All bases have the same size; this size is the {\em rank} of $M$.  A {\em circuit} $S$ of $M$ is a minimal dependent set of $M$, and $e \in E$ is a \emph{loop} if $\{e\}$ is dependent. For $S \subseteq E$, the {\em rank of $S$} in $M$ is the maximum size of an independent subset of $S$. Associated with $M$ is another matroid $M^*(G)$, called the dual matroid of $M$.  Its independent sets are $\{E \setminus I \mid I \in \cI\}$.
\end{definition}

A graph $G = (V, E)$ has an associated \emph{graphic matroid} $M(G)$ whose ground set is $E$, and whose independent sets are the subsets $S \subseteq E$ containing no cycles of $G$.  When $G$ is connected, the bases of $M(G)$ are the spanning trees of $G$, and the bases of the dual graphic matroid $M^*(G)$ are the complements of the spanning trees of $G$. For a connected graph $G = (V, E)$, the \emph{rank} of $G$ (the rank of $M(G)$) is $|V|-1$, and the \emph{corank} of $G$ is $|E| - |V| + 1$ which corresponds to the rank of the dual matroid $M^*(G)$.

\section{Adaptive querying strategies and the Tutte polynomial} \label{sec: tutte polynomial adaptive} 

The Tutte polynomial of a graph $G$ is a two-variable polynomial $T(G, x,y)$ which can be defined recursively as follows.  If $G$ contains only bridge edges and loops, then $T(G,x,y) = x^i y^j$, where $i$ is the number of bridge edges and $j$ is the number of loops. Otherwise, let $e$ be an edge of $G$ that is not a bridge edge or a self loop.  Then $T(G,x,y) = T(G/e,x,y) + T(G \setminus e,x,y)$. It is an essential property of the Tutte polynomial that the resulting polynomial $T(G,x,y) = \sum_{i, j} t_{i,j}x^i y^j$ will be the same, no matter which (non-bridge and non-loop) edge $e$ is chosen in the recursive definition.  Here $t_{i,j}$ is the coefficient in the resulting polynomial.
(This can be proved by induction on the number of edges in $G$).

Applying the recursive definition of the Tutte polynomial for a graph $G$ gives rise to a binary decision tree, with the first chosen edge $e$ as the query at the root, and recursively, trees for $G / e$ and $G \setminus e$ as the 0-child and 1-child of the root.  Call such a decision tree a Tutte strategy for $G$.

Each leaf $\ell$ of the binary decision tree then corresponds to a graph $G_{\ell}$ with only bridge edges and loops, and to a term $x^iy^j$ where $i$ is the number of bridges and $j$ is the number of loops.  
When $G$ is a connected graph, leaf $\ell$ also corresponds to a spanning tree $T_{\ell}$ of $G$ consisting of all bridge edges of $G_{\ell}$ and all edges $e$ contracted along the path from the root of the decision tree to $\ell$.   

The Tutte polynomial of $G$ is equal to the sum, over the leaves of this decision tree, of their corresponding $x^iy^j$.  Call such a decision tree a {\em Tutte strategy} for the graph $G$.  The value of each coefficient $t_{i,j}$ in the Tutte polynomial is equal to the number of leaves of the tree whose corresponding term is $x^iy^j$.

Given an instance $\mathcal{I} = (G=(V,E),p)$ of the MST-ID problem, and an edge $e \in G$, denote by $\mathcal{I}_{e=0}$ the instance $(G/e,p)$ of the MST-ID problem.  Similarly, denote by $\mathcal{I}_{e=1}$ the instance $(G\setminus e,p)$.  (In both cases, $p$ is restricted to the edges in $E\setminus\{e\}$).
Given an adaptive strategy $\mathcal{T}$ solving an MST-ID instance that begins by querying an edge $e$, denote by $\mathcal{T}^0$ and $\mathcal{T}^1$ the substrategies executed subsequently for $w_e=0$ and $w_e=1$ respectively. In what follows, when we refer to a strategy that solves an instance of MST-ID, that strategy may not be optimal.

\begin{lemma}
\label{lem:decompose}
Let $\mathcal{I}=(G,p)$ be an instance of the adaptive MST-ID problem.  If $\mathcal{T}$ is a strategy for $\mathcal{I}$ whose first query is on an edge $e$ that is neither a loop nor a bridge of $G$, then $\mathcal{T}^0$ and $\mathcal{T}^1$ are strategies solving the MST-ID problem on $\mathcal{I}_{e=0}$ and $\mathcal{I}_{e=1}$, respectively.  Conversely, if $e$ is an edge of $G$ that is neither a loop nor a bridge, and $S_0$ and $S_1$ are strategies solving the MST-ID problem for $\mathcal{I}_{e=0}$ and $\mathcal{I}_{e=1}$ respectively, then the following is a strategy solving the MST-ID problem for $\mathcal{I}$:  first query $e$, then execute $S_0$ if $w_e=0$ and $S_1$ if $w_e=1$. 
\end{lemma}

The next two lemmas follow easily from Lemma~\ref{lem:decompose}.

\begin{lemma}
\label{lem:optistutt}
Consider an instance $\mathcal{I}=(G,p)$ of the MST-ID problem. 
If $\mathcal{T}$ is an optimal strategy solving this instance, then $\mathcal{T}$ is a Tutte strategy.
\end{lemma}

\begin{lemma}
\label{lem:tuttesolves}
Let $\mathcal{I}=(G,p)$ be an instance of the MST-ID problem.  Let $\mathcal{T}$ be a Tutte strategy for $G$.  Then $\mathcal{T}$ solves the MST-ID problem for $\mathcal{I}$.
\end{lemma}

Now consider the following algorithm.  It executes an adaptive query strategy that solves the MST-ID problem for $G$.
The algorithm takes as input a fixed permutation $\pi$ of the edges of $G$.  It queries the edges of $G$ in the order specified by $\pi$, except that it skips any edges that are bridge or self loops in the minor of $G$ created by contracting edges $e$ already found to have $w_e=0$, and deleting edges $e$ already found to have $w_e=1$.  

\vspace{2mm}

{\bf \noindent ExecuteTutteStrategy($\pi$):}
\begin{enumerate}  
\item Initialize an MST $T$ of $G$ to be the empty set.
\item Process the edges $e$ of $G$ in the order specified by $\pi$ as follows:
If $e$ is a bridge edge 
of $G$, add it to $T$ without querying it.  If $e$ is a loop of $G$, do not query it.  Otherwise, query $e$.  If $w_e=0$ contract it in $G$ and add it to $T$.  If $w_e=1$, delete it from $G$.
\item Output $T$.
\end{enumerate}

Using the above, it is easy to prove the following theorem.  Although stated for a strategy based on a permutation $\pi$, the result holds for any strategy that processes an unqueried edge at each step, provided it  does not query edges that become loops and bridges of $G$.

\begin{theorem} \label{thm: all pe are half}
ExecuteTutteStrategy($\pi$) is an optimal algorithm for the adaptive MST-ID problem, for every permutation $\pi$ of the edges of $G$, in the special case where all $p_e=1/2$. 
\end{theorem}

\begin{proof}
Consider an instance $\mathcal{I}=(G,p)$ of the adaptive MST-ID problem, where all $p_e=1/2$. 
Let $\pi$ be a permutation of the edges of $G$.
Consider the Tutte strategy $\mathcal{T}$ for $G$ which results from considering the edges of $G$ in the order specified by $\pi$, skipping edges that become bridges and loops.
Thus $\mathcal{T}$ is the Tutte strategy executed by ExecuteTutteStrategy($\pi$).
By Proposition~\ref{prop:fundcyc}, the tree $T$ output by ExecuteTutteStrategy($\pi$) 
is certified as an MST of $G$ by the partial assignment $\beta \in \{0,1,*\}^E$ representing the answers received to the queries during execution.
It remains to show that $\mathcal{T}$ has minimum possible expected cost.

Consider a root-leaf path in $\mathcal{T}$.
Let $x^iy^j$ be the term corresponding to the leaf at the end of the root-leaf path.
 The number of queries performed on this path is $m-(i+j)$, where $m$ is the number of edges of $G$.
 When $p_e=1/2$ for all $e$, the probability of following this path is $(1/2^{m})2^{i+j}$.
 Thus the expected cost of the tree is equal to $\sum_{i,j} t_{i,j}(m-(i+j))(1/2^{m-(i+j)})$, where the $t_{i,j}$ are the coefficients of the Tutte polynomial.  Since the values of the $t_{i,j}$ do not depend on the ordering of the edges in the tree, it follows that the expected cost of the strategy is the same for any Tutte strategy $\mathcal{T}$ for $G$.  By Lemma~\ref{lem:optistutt}, there is a Tutte strategy $\mathcal{T}'$ for $G$ that is an optimal strategy for $\mathcal{I}$. 
 Since ExecuteTutteStrategy($\pi$) computes a Tutte strategy for $G$, and all such strategies have expected cost equal to the expected cost of $\mathcal{T}'$, it computes an optimal strategy for $\mathcal{I}$.
\end{proof}

The above proof relies heavily on the fact that the $p_e$ are all equal to 1/2.  In fact, the proof can be extended to the case where all $p_e$ are equal to another fixed value in $(0,1)$ (see Corollary~\ref{cor: all pe are equal} in Appendix~\ref{app: adaptive}).  However,
with arbitrary, unequal $p_e$, the optimal tree is still a Tutte tree, but its expected cost can differ depending on the order in which edges are processed. 

Using an amortized cost argument, we can extend the above ideas to get an approximation algorithm for the case where edge probabilities $p_e$ all lie on one side of $\frac{1}{2}$.
Let $DecP$ be the Tutte strategy that queries the edges in decreasing $p_e$ order, skipping created loops and bridges.  
That is, $DecP$ is executed by ExecuteTutteStrategy($\pi$) when $\pi$ lists the edges in decreasing $p_e$ order.
Let $IncP$ be the symmetric strategy that queries in increasing $p_e$ order. 

\begin{theorem} \label{thm :amortization}
$DecP$ is a 2-approximation for the special case of the MST-ID problem where
$p_e \le \frac{1}{2}$ for all edges $e$.   
$IncP$ is a 2-approximation for the special case of the MST-ID problem where
$p_e \ge \frac{1}{2}$ for all edges $e$.
\end{theorem}

\section{Algorithms for special cases of adaptive MST-ID}

We now present algorithms for other special cases of adaptive MST-ID.

\subsection{Easy graphs}
It is easy to solve the adaptive MST-ID problem for some very special graphs.  For example, if the graph is a simple cycle, querying must continue until either a weight 1 edge is identified, or $m-1$ weight 0 edges are identified.  It follows that an optimal adaptive strategy is actually nonadaptive: it is represented by a permutation of the queries, and querying continues in the order specified by the permutation until either an edge $e$ with $w_e=1$ is found, or the first $m-1$ edges in the permutation are found to have $w_e=0$. 
By an exchange argument, the first $m-1$ queries in the optimal permutation must be listed in  decreasing $p_e$ order.  The only question is which $m-1$ queries go first.  This can be determined by generating, for each edge $e$, the permutation which first lists the other $m-1$ edges, sorted in decreasing $p_e$ order, and then lists $e$.  
By computing the expected cost of the strategy corresponding to each of these lists,  the optimal adaptive (and non-adaptive) strategy is found.  

It follows that optimal strategies can also be easily computed for some related graphs. We refer to a graph obtained from a simple path (simple cycle) by replacing each edge with one or more parallel edges as a  \emph{multi-edge path}  (\emph{multi-edge cycle}). When $G$ is a single multi-edge, an optimal strategy can be computed by exchanging the roles of 0 and 1 in the above procedure for a simple cycle.  When $G$ is a multi-edge path, then finding an MST for $G$ requires finding an MST for each multi-edge. Thus, an optimal adaptive (and non-adaptive) strategy can be found by solving the MST-ID problem on each edge independently. However, the MST-ID problem is much more challenging for other types of graphs. We do not even know whether there is a polynomial-time exact algorithm for multi-edge cycles.

\subsection{Min-degree 3 graphs and simple planar graphs}
\label{sec:mindeg3}

We give a constant-factor approximation algorithm for graphs with no degree 2 vertices.  We then show how it can be used to obtain a constant-factor approximation algorithm for simple planar graphs. These approximation algorithms are based on the optimal solution of the MST-WD problem in which you not only need to identify the MST but also determine its weight.

\subsubsection{An optimal algorithm for MST-WD}
Let $G=(V,E)$ be the connected input graph, with independent edge weights $w_e\in\{0,1\}$ and $p_e=\Pr(w_e=1)$. We will assume w.l.o.g. that $0 < p_e < 1$, for all $e \in E$. A partial realization $\beta\in\{0,1,*\}^E$ records the queried weights, with $*$ denoting an unqueried edge. Define $
    G^0_\beta := (V,\{e\in E:\beta_e=0\}).
$
Our algorithm follows a greedy strategy. It maintains the connected components of $G^0_\beta$ and repeatedly queries an edge with minimum $p_e$ among those joining distinct components induced by weight 0 edges. 

\vspace{2mm} 
{\bf\noindent MST-WD ALG}
\begin{enumerate}
    \item Fix a permutation $\pi$ of $E$ in increasing order of $p_e$, breaking ties arbitrarily. Initialize $\beta_e=*$ for every $e\in E$ and $T_0=\varnothing$.

    \item Process the edges in the order specified by $\pi$. If the endpoints of the current edge $e$ belong to the same component of $G^0_\beta$, skip $e$. Otherwise, query $e$ and set $\beta_e=w_e$; if $w_e=0$, also add $e$ to $T_0$.

    \item Extend $T_0$ to a spanning tree $T$ using queried $1$-weight edges, and output $T$ and its weight $|T\setminus T_0|$.
\end{enumerate}

\begin{theorem} \label{thm: mst wd optimal}
MST-WD ALG is an optimal algorithm for MST-WD.
\end{theorem}

We leave the full proof of Theorem \ref{thm: mst wd optimal} to Appendix \ref{app: adaptive} and summarize the key ideas here. A partial realization $\beta$ determines the MST weight iff every unqueried edge lies within a single weight 0 connected component of $G^0_\beta$. To analyze query complexity, we consider the residual multigraph formed by contracting these weight 0 components. For fixed-order strategies, an exchange argument shows that edges should be processed in increasing order of $p_e$. An induction argument on the number of remaining non-loop edges extends this fixed-order optimality to all adaptive decision trees.

\subsubsection{Graphs of minimum degree at least $3$}
For an instance $I=(G,p)$, let $\OPT_{\mathrm{ID}}(I)$ and
$\OPT_{\mathrm{WD}}(I)$ denote the minimum expected numbers of
queries required for MST-ID and MST-WD respectively.

\begin{theorem}
\label{thm:min-degree-three}
MST-WD ALG is a polynomial-time $9$-approximation algorithm for
MST-ID on loopless graphs whose vertices all have degree at least $3$.
\end{theorem}

\begin{proof}
Let $I=(G,p)$ be an MST-ID instance on a loopless graph $G=(V,E)$ of minimum degree at least $3$, and set $n=|V|$. Running an optimal MST-ID strategy and then querying every unqueried edge of its output tree determines the MST weight using at most $n-1$ additional queries. Therefore,
\begin{equation}
\label{eq:wd-vs-id}
    \OPT_{\mathrm{WD}}(I)\le \OPT_{\mathrm{ID}}(I)+n-1.
\end{equation}

We next show that every terminating execution of a correct MST-ID strategy makes at least $n/8$ queries. Fix such an execution, let $\beta$ be its final partial realization, and let $T$ be its output tree. Let $S\subseteq V$ be the set of vertices that form singleton components of $G^0_\beta$. We consider two cases.

\bigskip
{\bf \noindent Case 1:} $|S|\leq \frac{3n}{4}$. Let $\mathcal{C}$ be the set of non-singleton components of $G^0_\beta$. Each $C\in\mathcal{C}$ contains at least $|C|-1\geq |C|/2$ queried $0$-weight edges. Therefore, the total number of queries is at least
\[
\sum_{C\in \mathcal{C}} (|C|-1) \geq \frac{1}{2} \sum_{C\in \mathcal{C}} |C| = \frac{1}{2} (n-|S|) \geq \frac{n}{8}.
\]

{\bf\noindent Case 2: $|S|>\frac{3n}{4}$.} Let $E_S\subseteq E$ be the set of edges with at least one endpoint in $S$. Each edge in $E_S$ contributes at most $2$ to the sum of the degrees of the vertices in $S$, so
\[
|E_S|\geq \frac{1}{2}\sum_{v\in S}\deg_G(v)\geq \frac{3}{2}|S|>\frac{9n}{8}.
\]
Since $G$ is loopless and every vertex in $S$ forms a singleton component of $G^0_\beta$, every edge in $E_S$ joins distinct components of $G^0_\beta$. If an edge of $E_S$ were unqueried and outside $T$, Proposition~\ref{prop:fundcyc} would place both endpoints in the same component of $G^0_\beta$, whereas every edge of $E_S$ joins distinct components. Thus every edge of $E_S\setminus T$ has been queried. Since $T$ contains only $n-1$ edges, at least $|E_S|-(n-1) \geq \frac{9n}{8}-(n-1)\geq \frac{n}{8}$
edges of $E_S$ must have been queried. Every correct MST-ID strategy makes at least $n/8$ queries on each execution, so its expected query cost is at least $n/8$. Taking the minimum over all such strategies gives $\OPT_{\mathrm{ID}}(I)\ge n/8$. Hence $n-1\leq 8\OPT_{\mathrm{ID}}(I)$, and \eqref{eq:wd-vs-id} yields  $\OPT_{\mathrm{WD}}(I) \leq \OPT_{\mathrm{ID}}(I)+n-1 \leq 9\OPT_{\mathrm{ID}}(I).$

MST-WD ALG outputs an MST, so it is also a valid algorithm for MST-ID. By the preceding theorem, its expected query cost is $\OPT_{\mathrm{WD}}(I)$. The displayed inequality therefore proves the claimed $9$-approximation guarantee. Finally, it is easy to see that sorting the edges, maintaining connected components, and extending the forest to a spanning tree all take polynomial time.
\end{proof}

\subsubsection{Simple planar graphs}
We obtain the same approximation guarantee for simple planar graphs by applying Theorem~\ref{thm:min-degree-three} to a suitable dual instance. The reduction first contracts bridges and then complements the edge weights.

\begin{theorem}
\label{thm:simple-planar}
There is a polynomial-time $9$-approximation algorithm for MST-ID on
simple planar graphs.
\end{theorem}

\begin{proof}
Let $I=(G,p)$ be an MST-ID instance where $G$ is a simple planar connected graph. If $G$ is a tree, output $G$ without making any queries. Otherwise, contract all bridges of $G$. Let $H$ be the resulting graph, and let $p_H$ denote the restriction of $p$ to $E(H)$. Every spanning tree of $G$ contains all bridges, whose weights contribute the same amount to every spanning tree. Contracting the bridges gives a bijection between spanning trees of $G$ and $H$, so $\OPT_{\mathrm{ID}}(H,p_H)=\OPT_{\mathrm{ID}}(G,p)$, and it also preserves simplicity and planarity.

Thus $H$ is connected, simple, planar, and bridgeless, with at least one edge, and every edge of $H$ lies on a cycle. Fix a planar embedding of $H$, and let $H^*$ be its planar dual. Since $H$ is connected and has no bridges, $H^*$ is connected and loopless. Moreover, since $H$ is simple, every face boundary has length at least $3$. This boundary length equals the degree of the corresponding dual vertex, so every vertex of $H^*$ has degree at least $3$.

For each $e\in E(H)$, let $e^*$ be its corresponding dual edge in $H^*$.
We assign complementary weights to these edges: $e^*$ has weight $1$
when $e$ has weight $0$, and weight $0$ when $e$ has weight $1$.
Since $p_e$ is the probability that $e$ has weight $1$, the probability
that $e^*$ has weight $1$ is $p^*_{e^*}=\Pr(w_e=0)=1-p_e.$

For a fixed realization $\alpha$ of the weights on $H$, the corresponding
dual realization $\alpha^*$ is obtained by applying this rule to every
edge: $\alpha^*_{e^*}=1-\alpha_e$ for all $e\in E(H)$. The corresponding realizations have the same probability under the two distributions. A query to $e^*$ is simulated by querying $e$ and complementing the answer, and conversely.

For a spanning tree $T$ of $H$, define $T^*=\{e^*:e\in E(H)\setminus T\}$. By planar duality, $T^*$ is a spanning tree of $H^*$, and this gives a bijection between spanning trees of $H$ and spanning trees of $H^*$. 
The map sends $T$ to the duals of the edges omitted from $T$, so
corresponding trees may have different numbers of edges. Euler's formula gives $|T^*|=|E(H)|-(|V(H)|-1)=|V(H^*)|-1$, and the inverse map recovers $T$ by taking the complementary primal edges.

Now writing $\alpha(T)$ and $\alpha^*(T^*)$ for their respective total weights, we have
\[
\begin{aligned}
    \alpha^*(T^*) = \sum_{e\in E(H)\setminus T}(1-\alpha_e) 
    = |T^*|-\sum_{e\in E(H)\setminus T}\alpha_e 
    = |V(H^*)|-1-\sum_{e\in E(H)}\alpha_e+\alpha(T).
\end{aligned}
\]
Since all terms in the final expression except $\alpha(T)$ are independent of $T$, a spanning tree $T$ is an MST under $\alpha$ iff $T^*$ is an MST under $\alpha^*$.

Complementing the queried weights and leaving unqueried edges unqueried gives a bijection between the completions of corresponding partial realizations. Thus $T$ is an MST for every completion exactly when $T^*$ is. Since queries can be simulated in either direction with the same cost and matching outcome probabilities, we obtain $\OPT_{\mathrm{ID}}(H^*,p^*)=\OPT_{\mathrm{ID}}(H,p_H)$.

We now run MST-WD ALG on the dual instance $(H^*,p^*)$, simulating each dual query by one query to the corresponding edge of $G$. Let $T^*$ be its output tree and $T=\{e\in E(H):e^*\notin T^*\}$. Output the original edges of $G$ corresponding to $T$, together with all contracted bridges. The preceding correspondences show that this is an MST of $G$. By Theorem~\ref{thm:min-degree-three}, the expected number of queries is at most $9\OPT_{\mathrm{ID}}(H^*,p^*) =9\OPT_{\mathrm{ID}}(H,p_H) =9\OPT_{\mathrm{ID}}(G,p).$
Finding and contracting bridges, constructing a planar embedding and its dual, and running MST-WD ALG all take polynomial time.
\end{proof}

\section{Hardness}
\label{sec:hardness}

\begin{lemma} \label{lemma: expected cost of tutte}
    Consider an instance $\mathcal{I}=(G,p)$ of the MST-ID problem where $G=(V,E)$ and there exists $q \in (0,1)$ such that $q=1-p_e$ for all $e \in E$.
    Let $m=|E|$, and $n=|V|$.
    Let $t_{i, j}$ denote the coefficient of $x^iy^j$ in the Tutte polynomial representing graph $G$.
     
    The expected cost of a Tutte strategy solving the MST-ID problem for instance $\mathcal{I}$ as a function of $q$, is:

    $$C(q) := q^{n-1}(1-q)^{m-(n-1)}\sum_{i,j} t_{i,j} q^{-i}(1-q)^{-j}(m-(i+j))$$

    \noindent Furthermore, define $T_D(x) := T(x, x) = \sum_{i, j} t_{i, j} x^{i+j}$ to be the diagonal Tutte polynomial. Then 
    $$C(\tfrac{1}{2}) = m - (\tfrac{1}{2})^{m-1} T_D'(2)$$
    where $T_D'(x)$ is the derivative of $T_D(x)$.  So $C(\frac{1}{2})$ is computable in polynomial time iff $T_D'(2)$ is.
\end{lemma}

\begin{theorem} \label{thm: sharp P hard}
    The problem of computing the expected cost of an optimal solution to an MST-ID instance is 
     $\#P$-hard, even in the special case where all $p_e = \frac{1}{2}$.
\end{theorem}

The proof of Theorem \ref{thm: sharp P hard} is by reduction from two-terminal network reliability \cite{valiant1979complexity} and using Lemma \ref{lemma: expected cost of tutte}.

\section{Non-adaptive MST identification}
\label{sec:non-adaptive}

\subsection{An $O(\log E)$ approximation algorithm for general graphs}

We give an $O(\log E)$ approximation algorithm solving the non-adaptive MST-ID problem for general graphs.  
Our algorithm is based on a technique recently introduced by Brody et al.\ \cite{brody2026extensionsunanimousvoteproblem}.
The technique produces an approximately optimal non-adaptive querying strategy for a (certain type of) querying problem by reducing the problem to Submodular Ranking.  The resulting  Submodular Ranking instance is then solved (approximately), using a generic modified version of the greedy algorithm of Azar and Gamzu for that problem~\cite{AzarGamzu} in which the greedy decisions are based on estimates derived from random sampling. We begin by defining the Submodular Ranking problem.

\subsection{Submodular Ranking}

The input to the Submodular Ranking problem is a set of $m$ monotone submodular set functions $f_1,\ldots,f_m$ defined on subsets of a {\em ground set} $A$, along with weights $w_i \in \mathbb{R}_{\geq 0}$ for $i \in [m]$.
For each $i \in [m]$, $f_i:2^{A} \rightarrow [0,1]$, $f_i(\emptyset)=0$ and $f_i(A) = 1$.

Let $n = |A|$. Given a permutation $\pi:\{1, \dots, n\} \rightarrow A$ of the elements of $A$, the \emph{cover time} of $f_i$ in $\pi$, denoted $t(f_i, \pi)$, is the smallest integer $j \in \{1, \dots, n\}$ such that $f_i(\{\pi(1), \dots, \pi(j)\})=1$.  The Submodular Ranking problem asks for a permutation $\pi:\{1, \dots, n\} \rightarrow A$ with minimum cost, where the cost is the weighted sum of the cover times of the $f_i$ in $\pi$: 
 
\begin{equation}
\text{cost}(\pi) := \sum_{i=1}^m w_i\cdot t(f_i,\pi)
\end{equation}

The greedy approximation algorithm of Azar and Gamzu for Submodular Ranking works as follows. It builds a permutation $\pi$ of the elements, starting from an initially empty permutation, by repeatedly  choosing an element $j \in A \setminus S$ maximizing

\begin{equation}
\text{score}(S,j) :=\sum_{i:f_i(S) < 1} \frac{f_i(S \cup \{j\}) - f_i(S)}{1 - f_i(S)}.
\end{equation}
where $S$ is the set of already chosen items. The algorithm achieves an $O(\log 1/\epsilon)$-approximation factor: its constructed permutation has cost within a factor $O(\log 1/\epsilon)$ of the cost of an optimal permutation~\cite{AzarGamzu,Im2016}.  Here $\epsilon$ is the minimum non-zero value of $f_i(S \cup \{j\})-f(S)$, over all $i \in [m]$ and $S \subset A$.
Thus $\epsilon$ is the minimum non-zero marginal increase for any of the $f_i$.

\subsection{Reduction to Submodular Ranking}

To reduce the MST-ID problem to Submodular Ranking, we take an instance $(G=(V,E),p)$ of MST-ID, and use $E$ as the groundset for the reduced instance of Submodular Ranking. 

 For each $\alpha \in \{0,1\}^E$,  
 we define a function $f_{\alpha}:2^E \rightarrow [0,1]$ such that $f_{\alpha}$ is monotone and submodular,
$f_{\alpha}(\emptyset) = 0$, and $f(E) = 1$.  Additionally,
for each $S \subseteq E$, $f_{\alpha}(S) = 1$ iff $\alpha^S$ certifies some MST $T$ of $G$.  
Finally, given any $S \subseteq A$, $f_{\alpha}(S)$ can be computed in polynomial time.
We defer the details of the definition of $f_{\alpha}$ to Section~\ref{sec: goal value construction}. This definition is specific to the MST-ID problem (and not given by the technique of~\cite{brody2026extensionsunanimousvoteproblem}).  The $2^n$ functions $f_{\alpha}$ are the functions $f_i$ in the reduced Submodular Ranking instance.

Let $w_{\alpha}$ be the probability of $\alpha$ in the MST-ID instance $(G,p)$.
Thus given a permutation $\pi$ of the elements of $E$, 
$cost(\pi)$ is equal to the expected cost of non-adaptive strategy $\pi$ for MST-ID instance $(G,p)$.   It follows that applying the greedy algorithm for Submodular Ranking to the instance created by the reduction yields a permutation $\pi$ that is an $O(\log 1/\epsilon)$-approximation to the optimal non-adaptive strategy for MST-ID instance $(G,p)$.  We show in Section~\ref{sec: goal value construction} that for our definition of the $f_{\alpha}$, $\frac{1}{\epsilon}$ is polynomial in $|E|$, so the approximation factor is logarithmic in $|E|$.  

However, running the greedy algorithm for Submodular Ranking on the exponentially many $f_{\alpha}$ is not computationally feasible. 
Brody et al.\ \cite{brody2026extensionsunanimousvoteproblem} showed how to overcome this difficulty, asssuming we can sample from the distribution on realizations $\alpha$ (as is the case for the MST-ID problem).  In each greedy step of the Submodular Ranking algorithm, we can instead use sampling to estimate the $\text{score}(S,j)$ values, and run the greedy algorithm using the estimated scores.  With a sample of size polynomial in $|A|$ in each greedy step, the expected cost of the permutation produced by the greedy algorithm is still within a factor of $O(\log 1/\epsilon)$ of optimal.  

\subsection{Definition of the goal functions $f_{\alpha}$} \label{sec: goal value construction}

Fix a realization $\alpha \in \{0, 1\}^E$.
We will define a monotone submodular function
$g_{\alpha} : 2^E \to \{0,1,2,\ldots, Q_{\alpha}\}$ 
such that $g_\alpha(S) = Q_\alpha$ iff $\alpha^S$ is an MST-ID certificate for $G$.  
We define $f_{\alpha}:2^E \to [0,1]$ such that $f_{\alpha}(S) = \frac{1}{Q_{\alpha}}g_{\alpha}(S)$ for all $S \subseteq E$.
Thus $f_{\alpha}(S)=1$ iff $\alpha^S$ is an MST-ID certificate for $G$.

The definition of $g_{\alpha}$ is based on two different minors of $G$, $G^*=(V^*,E^*)$ and $G'=(V',E')$, which we define here.  $G^*$ is the minor of $G$ produced by contracting all weight 0 edges of $G$ and deleting all resulting loops.  
Thus the vertices of $G^*$ correspond to 
$C^* :=$ 
the set of connected components of the subgraph of $G$ containing only its weight 0 edges.  We call the $c^* \in C^*$ the {\em weight 0 connected components} of $G$.   
Graph $G'$ is produced by contracting just the weight 0 edges of $G$ that are on a weight 0 cycle of $G$ and deleting all resulting loops.  Thus the vertices of $G'$ correspond to $C':=$ the set of connected components of the subgraph of $G$ containing only the edges on weight 0 cycles of $G$.  The $c' \in C'$ are the 2-edge-connected components of the subgraph of $G$ consisting only of the weight 0 edges. Note that each $c' \in C'$ is contained in some unique $c^* \in C^*$ (See Figure \ref{fig:goalvalue-example} for an example of these constructions) .  Since each $c' \in C'$ is a graph, we write $c' = (V_{c'},E_{c'})$; similarly for $c^* \in C^*$ we write $c^* = (V_{c^*},E_{c^*})$. 

Consider a weight 0 connected component $c^* \in C^*$. The edges of $c^*$ present in $G'$ (all of which have weight 0) form a tree in $G'$, as $G'$ is formed by contracting all the weight 0 cycles of $G$. This tree, which we denote by $t_{c^*}$, spans all vertices of $G'$ corresponding to the 2-edge-connected components $\{c' \in C' \mid c' \subseteq c^*\}$. For each $c^* \in C^*$, we construct a bipartite graph $B_{c^*}$ with vertex set $L \sqcup R$, where $L$ is the set of edges in $t_{c^*}$, and $R$ is the set of weight 1 edges in $G'$ that connect two different vertices of $t_{c^*}$. An edge in $B_{c^*}$ connects $e_L \in L$ and $e_R \in R$ iff in $G'$, $e_L$ lies on the fundamental cycle that $e_R$ forms with $t_{c^*}$. Because this fundamental cycle consists solely of $e_R$ and edges of $t_{c^*}$ (which are in $L$), every edge in $B_{c^*}$ joins a vertex in $R$ to a vertex in $L$, making $B_{c^*}$ bipartite.

Our construction of $g_{\alpha}$ is based on the 
the following lemma.  It gives necessary and sufficient conditions for a set of edge queries, answered according to the fixed full realization $\alpha$, to constitute a certificate of an MST of $G$.

\begin{lemma}
\label{lem:nacert}
Let $S \subseteq E$ and $\alpha \in \{0,1\}^E$. Then $\alpha^S$ is a certificate of an MST $T$ of $G$ iff there exist three disjoint subsets of $S$, call them $S_1, S_2, S_3$ such that:
\begin{itemize}
\item[(1)] $S_1$ is a set of edges in $G^*$ such that $E^*\setminus S_1$ is a spanning tree of $G^*$.  
\item[(2)] $S_2 = \bigcup_{c' \in C'} T_{c'}$ where each $T_{c'}$ is a spanning tree of $c'$. 
\item[(3)] $S_3 = \bigcup_{c^* \in C^*} D_{c^*}$, where $D_{c^*}$ is a set of edges of $G$ that correspond to vertices in $B_{c^*}$ that form a cover of the edges of $B_{c^*}$.
\end{itemize}
\end{lemma}

We now describe the construction of $g_{\alpha}$.
Function $g_\alpha$ is the sum of 3 terms: $g_\alpha = g^* + \sum_{c'} g_{c'} + \sum_{c^*} g_{c^*}$, where $g^*$, the $g_{c'}$, and the $g_{c^*}$ are all monotone submodular functions. We describe each of them below.  In what follows, we rely on the definitions of $G^*$,$G'$, $C'$, and $C^*$ given above. 

{\bf \noindent The Function $g^*$: }Let $M^*(G^*)$ be the dual matroid of $M(G^*)$. 
For $S \subseteq E$, define
$g^*(S) $ to be 
the rank of $S \cap E^*$ in the dual matroid of $G^*$, $M^*(G^*)$.
This function is monotone and submodular (as a rank function of a matroid). 
Define $Q^* :=  |E^*| - |V^*| + 1$.
Because $G$ is connected, $Q^*$ is the rank of the dual matroid $M^*(G^*)$. 

{\bf \noindent The functions $g_{c'}$:} Let $G_0$ be the subgraph formed by weight 0 edges under $\alpha$. Let $c' \in C'$ be a $2$-edge-connected component of $G_0$ (the components $C'$ can be found in linear time using Tarjan's bridge-finding algorithm~\cite{tarjan1974bridges}). 
Recall that $M(c')$ is the matroid associated with graph $c'$. For $S \subseteq E$, 
define $g_{c'}(S)$ to be the rank of 
$S \cap E_{c'}$ in $M(c')$. Define $Q_{c'} := |V_{c'}| - 1$. Thus $Q_{c'}$ is the size of a spanning tree of $c'$.

{\bf \noindent The functions $g_{c^*}$:} Let $c^* \in C^*$.
Define $g_{c^*}(S)$= (number of edges of $B_{c^*}$ covered by vertices corresponding to edges in $S$). That is, edge $(e_0, f)$ in $B_{c^*}$ is covered by $S$ if $e_0 \in S$ or $f \in S$. $g_{c^*}$ is monotone and submodular as it is a coverage function. Define $Q_{c^*} := \text{ number of edges in } B_{c^*}$

\paragraph{Definition of $g_{\alpha}$: } We define $g_\alpha$ as: $g_\alpha := g^* + \sum_{c'} g_{c'} + \sum_{c^*} g_{c^*}$
Also, define $Q_\alpha := Q^* + \sum_{c'} Q_{c'} + \sum_{c^*} Q_{c^*}$.
We note that for $S=E$, each individual function reaches its goal value: $g^*(E) = Q^*$, $g_{c'}(E) = Q_{c'}$, and $g_{c^*}(E) = Q_{c^*}$, so $g_{\alpha}(E) = Q_\alpha$. Figure \ref{fig:goalvalue-example} (Appendix \ref{app: non-adaptive}) illustrates the definition of $g_{\alpha}$ and $Q_\alpha$ on a small example graph.

\begin{lemma}
  \label{prop:certificate if goal value}
  $g_\alpha(S) = Q_\alpha$ iff $\alpha^S$ is a  certificate for an MST $T$ of $G$. Moreover, $Q_\alpha = O(E^2)$.
\end{lemma}

Lemma \ref{prop:certificate if goal value} follows from Lemma \ref{lem:nacert} and the definitions of $g^*$, $g_{c'}$, and $g_{c^*}$. Applying the greedy algorithm to this goal function construction yields the following theorem.

\begin{theorem} \label{thm: nonadaptive main theorem}
    There is a $O(\log |E|)$-approximation algorithm that solves the non-adaptive MST-ID problem.
\end{theorem}

\section{Submodular goal values for the adaptive case}
\label{sec:goalvalue}

In this section, we explore the possibility of using the {\em goal value approach} of Deshpande et al. to develop approximation algorithms for adaptive MST-ID problems.  We need the following definitions.

Let $\Omega$ be a finite set. 
For $\beta \in (\Omega \cup \{*\})^I$, $i \in I$, and $s \in \Omega$, we use $\beta_{i \leftarrow s}$ to denote the element of $(\Omega \cup \{*\})^I$ produced from $\beta$ by changing $\beta_i$ to $s$. Let $g:(\Omega \cup \{*\})^I \rightarrow \mathbb{R}_{\geq 0}$. 
For $\beta \in (\Omega \cup \{*\})^I$ and $i \in I$ such that $\beta_i \neq *$, we consider $\beta_i$ to be the state of $i$. Thus $g(\beta)$ can be considered to be the value of the items $i \in I$ such that $\beta_i \neq *$, when each such item $i$ is in state $\beta_i$. Therefore we call $g$ a \emph{state-dependent} set function. We say that $g$ is \emph{monotone} if for all $\beta,\gamma \in (\Omega \cup \{*\})^I$ with $\beta \preceq \gamma$, we have $g(\beta) \leq g(\gamma)$. We say that $g$ is \emph{submodular} if for all $\beta,\gamma \in (\Omega \cup \{*\})^I$ with $\beta \preceq \gamma$, $i \in I$ such that $\gamma_i = *$, and $s \in \Omega$, we have $g(\beta_{i \leftarrow s}) - g(\beta) \geq g(\gamma_{i \leftarrow s}) - g(\gamma)$. 

Applying the goal value approach of Deshpande et al.\ to the
 adaptive MST-ID problem requires reducing it to the problem of Stochastic Submodular Cover.  For an instance $(G,p)$ of adaptive MST-ID, that reduction must construct a single, state-dependent monotone submodular set function $g \colon \{0,1,*\}^E \to \{0,1,2,\ldots,Q\}$ such that $g(*,*,*,\ldots,*) = 0$ and for all partial assignments $\beta \in \{0,1,*\}^E$, $g(\beta) = Q$ iff $\beta$ contains an MST-ID certificate for $G$.  We call $g$ a {\em submodular goal function} for the MST-ID problem on $G$, and $Q$ its goal value.  Running the standard greedy algorithm for Stochastic Submodular Cover problem on $g$ then gives an $O(\log Q)$-approximation to the optimal solution of the original adaptive MST-ID instance.  Thus to achieve an $O(\log E)$ approximation algorithm for adaptive MST-ID using this approach, there must exist a submodular goal function $g$ for $G$ whose goal value $Q$ is polynomial in $|E|$.  We leave as an open question whether such a function exists, but present some related observations and results.

First, we 
note that in Section \ref{sec:non-adaptive}, we were able to achieve an $O(\log E)$ approximation algorithm for the non-adaptive problem by
constructing a function $g_\alpha:2^E \rightarrow \{0,1,2,\ldots,Q_{\alpha}\}$ for each full realization $\alpha$.
Each such $g_{\alpha}$ was monotone and submodular, and satisfied the properties that $g(\emptyset)=0$, and for all $S \subseteq E$, $g(S) = Q_{\alpha}$ iff $\alpha^S$ was an MST-ID certificate for $G$.
We call such a function a submodular goal function (for the MST-ID problem) for $(G,\alpha)$, and $Q_{\alpha}$ its goal value.  Crucially, each $Q_{\alpha}$ was polynomial in $|E|$.

As observed by  Brody et al. \cite{brody2026extensionsunanimousvoteproblem},  if there exists a submodular goal function $g$ with goal value $Q$ for an adaptive problem, then for the corresponding non-adaptive problem,
there exists a submodular goal function $g_{\alpha}$ with goal value $Q_{\alpha}=Q$ 
for each full realization $\alpha$
(just take $g_{\alpha}(S) = g(\alpha^S)$). However, we show the converse does not hold in general. A graph can satisfy $\max_\alpha Q_\alpha \leq Q$ across individual realizations, yet require a strictly larger goal value in the adaptive setting. We give more detail on this below.

We define the submodular goal value for MST-ID of a graph $G$, denoted by $Q_{ID}(G)$, to be the minimum value $Q$ such that there is a submodular goal function with goal value $Q$ for the MST-ID problem on $G$. It's easy to see that the $Q_{ID}(G)$ is the same for a graph and its dual.  
Using observations from Bach et al.~\cite{bach2018submodular}, who studied the submodular goal value of Boolean functions, it is easy to show that for any graph $G$, $Q_{ID}(G)$ is lower bounded by both the rank and the corank of $G$, and is upper bounded by $2^{|E|}$. We refer to the product $(|V|-1)(|E|-|V|+1)$ as the \emph{rank$\times$corank value} of $G$. For some particular graphs $G$, we were able to determine the exact value of $Q_{ID}(G)$. In Appendix \ref{app: goal values}, we show that for multi-edge cycles and their dual graphs, the rank$\times$corank value is equal to $Q_{ID}(G)$.

The rank$\times$corank lower bound is not always tight. For example, the diamond graph ($K_4$ minus a diagonal) with one non-diagonal edge doubled, denoted $D^{+1}$, has a rank$\times$corank value of 9, but $Q_{ID}(D^{+1}) =10$. We computed $Q_{ID}$ for a variety of very small graphs, using Integer Linear Programming over state spaces of partial realizations. A table of the computed values is in Appendix~\ref{app: goal values}.

\section{Open problems}
  We left open the following main questions: (1) whether the adaptive MST-ID problem has a poly-time $O(\log E)$-approximation algorithm and whether there is a constant-factor approximation for the non-adaptive problem  (2)  whether the MST-ID problem is NP-hard and  (3) whether $Q_{ID}$ is polynomial in $|E|$.
   
   There are also natural questions to ask about extensions of the problems in this paper.
   We focused exclusively on binary edge weights and uniform query costs.   
   Introducing general edge weights and non-uniform query costs creates new open questions. 
Some of our approaches and techniques extend easily to handle non-uniform query costs.  For example, a natural variant of the greedy submodular ranking algorithm achieves the same $O(\log 1/\epsilon)$ approximation factor even when the elements of $j \in A$ each have a weight (cost $c_j > 0$) and the problem is to minimize the sum of the weighted covering times; the score function is just modified by dividing by $c_j$~\cite{Im2016}.  Thus our $O(\log E)$ approximation algorithm can easily be adapted to handle non-uniform query costs.  However,
handling nonbinary weights appears to be significantly more challenging.  

\section*{LLM usage disclosure}

Google Gemini was consulted regarding the complexity of computing the diagonal of the Tutte polynomial at $x=2$. It then provided the essential ideas for the proofs of Lemma~\ref{lem: computable if S(G) is} and Lemma~\ref{lem:computing S(G) is hard}.  It was also consulted in an unsuccessful attempt to improve some submodular goal value bounds presented in the paper.
It was used to
assist in converting hand drawn images to tikz format for Figure \ref{fig:goalvalue-example} (Appendix \ref{app: non-adaptive}). Claude was used to convert experimental results in Appendix (Table \ref{tab: goal values}) to Overleaf table format. LLMs did not write any part of this paper and the authors take full responsibility for it.

\section*{Acknowledgements}
We thank Evan Brody for his careful proofreading.

\bibliography{main}
\newpage
\appendix

\section{Deferred results and proofs for adaptive strategies} \label{app: adaptive}

\begin{proof} [Proof of Lemma~\ref{lem:decompose}]
Let $E$ be the edge set of $G$.
\paragraph{$\Rightarrow$}  Consider a path from the root of $\mathcal{T}$ down to a leaf.  Let $\beta$ be the partial assignment representing the answers received to the queries along that path.  
Since $\mathcal{T}$ solves the MST-ID problem, $\beta$ is a certificate for an MST $T$ of $G$.  Consider such a tree $T$.  Thus $\beta$ and $T$ satisfy the fundamental cycle condition in Proposition \ref{prop:fundcyc}.
Let $\beta'$ denote the restriction of $\beta$ to $E \setminus \{e\}$.  

Suppose
$\beta_e = 0$. (The case where $\beta_e=1$ is dual and we omit it.) To show that $\mathcal{T}^0$ solves the MST-ID problem for $\mathcal{I}_{e=0}$, it suffices to show that $\beta'$ is an MST-ID certificate for $G / e$.

Case 1: $e \in T$.  Then $T' := T \setminus \{e\}$ is a spanning tree of $G/e$.  
Consider an edge $e' \in E \setminus T$. Then the edges in its fundamental cycle with $T'$ in $G/e$ are the same as the edges in its fundamental cycle with $T$ in $G$, except for the possible absence of $e$.  Therefore, since $\beta$ and $T$ satisfy the fundamental cycle condition of Proposition \ref{prop:fundcyc} for graph $G$, $\beta'$  and $T'$ satisfy that condition for graph $G/e$.  Thus $\beta'$ certifies $T'$ 
 as an MST of $G/e$. 
 
Case 2: $e \not\in T$.   Consider the fundamental cycle $C_e$ that $e$ forms with $T$ in $G$.  By Proposition \ref{prop:fundcyc}, all $e' \in C_e \setminus \{e\}$ are assigned weight 0 by $\beta$.
Let $\hat{e}$ be an arbitrary edge in $C_e \setminus \{e\}$ and 
let $T'$ be the spanning tree of $G$ produced by deleting $\hat{e}$ and replacing it with $e$.  Then $T'' := T'\setminus \{e\}$ is a spanning tree of $G/e$.
We now show that $\beta'$ certifies $T''$ as an MST of $G/e$, using Proposition \ref{prop:fundcyc}.
Let $f \neq e$ be an edge in $G/e$ that is not in $T''$, such that $\beta_f =0$ or $\beta_f = *$. 
Let $C''_f$ be the fundamental cycle that $f$ forms with $T''$ in $G / e$.   
If $f=\hat{e}$, then $C''_f$ is equal to $C_e \setminus \{e\}$, all of whose edges are assigned weight 0 by $\beta$.  Now suppose $f \neq \hat{e}$.  Then for every edge $e'' \in C''_f \setminus \{f\}$, $e''$ must either be a member of the original fundamental cycle $C_f$ that $f$ forms with $T$ in $G$, or a member of $C_e \setminus \{\hat{e},e\}$.   We already know that the edges in $C_e \setminus \{\hat{e},e\}$ are assigned weight 0 in $\beta$, and by Proposition \ref{prop:fundcyc} applied to $\beta$ and $T$, the edges in $C_f \setminus \{f\}$ are also all assigned weight 0 in $\beta$.
 Thus $\beta'$ certifies $T''$ 
 as an MST of $G/e$.

\paragraph{$\Leftarrow$}
Let $e$ be a non-bridge and non-loop edge of $G$.
Suppose $S_0$ and $S_1$ are strategies solving the MST-ID problem for $\mathcal{I}_{e=0}$ and $\mathcal{I}_{e=1}$ respectively.  Consider the strategy $S$ that first queries $e$ and then executes $S_0$ if $w_e=0$ and $S_1$ if $w_e$ = 1.

Consider a root-leaf path in $S$ and let $\beta$ be the partial assignment representing the answers to the queries along the path.  
We will show that $\beta$ is a certificate for an MST $T$ of $G$.  Let $\beta'$ be $\beta$ restricted to $E\setminus \{e\}$.
Suppose $\beta_e = 0$. 
(The case where $\beta_e=1$ is dual and we omit it.) 
Since $S_0$ solves the MST-ID problem for $\mathcal{I}_{e=0}$, $\beta'$ certifies an MST $T'$ of $G/e$.  Consider $T' \cup \{e\}$.  Because $G/e$ is produced by contracting $e$ in $G$ (which merges the endpoints of $e$), and $T'$ is a spanning tree of $G / e$, $T:=T' \cup \{e\}$ is  produced by ``unmerging'' the endpoints of $e$ in $T$.  Thus $T$ is a spanning tree of $G$.  For each edge $e'$ in $E \setminus T$, the fundamental cycle it forms with $T$ has the same edges as the fundamental cycle it forms with $T'$, except perhaps it also has $e$.  Since $\beta'$ and $T'$ must satisfy the condition in Proposition \ref{prop:fundcyc} for graph $G/e$, and $\beta_e = 0$, it follows that $\beta$ and $T$ satisfy the condition in Proposition \ref{prop:fundcyc} for graph $G$.  Therefore $S$ is a solution to the MST-ID problem for $\mathcal{I}$.
\end{proof}

\begin{proof}[Proof of Lemma~\ref{lem:optistutt}]
~~~
The proof is by induction on the number of edges in $G$.  The statement is clearly true for graphs with 1 edge.
We show the inductive step as follows.
If $G$ has only bridges and loops, the lemma clearly holds.
Otherwise, consider an optimal strategy $\mathcal{T}$ for $\mathcal{I}$ and let the first query be on edge $e$.  
It is easy to show that if $e$ is a loop or bridge of $G$, then $\mathcal{T}^0$ (or $\mathcal{T}^1$) would also be a strategy for $\mathcal{I}$, contradicting the optimality of $\mathcal{T}$. 
By Lemma~\ref{lem:decompose}, $\mathcal{T}^0$ and $\mathcal{T}^1$ are strategies solving the MST-ID problem for $\mathcal{I}_{e=0}$ and $\mathcal{I}_{e=1}$.  Further, they must be optimal strategies for that problem, otherwise  they could be replaced in $\mathcal{T}$ by strategies with lower expected cost, contradicting the optimality of $\mathcal{T}$.  Therefore, by the inductive hypothesis, $\mathcal{T}^0$ and $\mathcal{T}^1$ are Tutte strategies for $G/e$ and $G \setminus e$.   Therefore, $\mathcal{T}$ is a Tutte strategy for $G$.
\end{proof}

\begin{proof}[Proof of Lemma~\ref{lem:tuttesolves}]
~~~
The proof is by induction on the number of edges in $G$.  The statement is clearly true for graphs with 1 edge.
We show the inductive step as follows.
If $G$ has only bridges and loops, the lemma clearly holds.

Otherwise, consider a Tutte strategy $\mathcal{T}$ for $G$ and let the first query be on edge $e$.  Since $\mathcal{T}$ is a Tutte strategy, $e$ is not a bridge or a loop. By induction, $\mathcal{T}^0$ is a solution to the MST-ID problem for instance $\mathcal{I}_{e=0}$, and $\mathcal{T}^1$ is a solution to the MST-ID problem for instance $\mathcal{I}_{e=1}$. Then by Lemma \ref{lem:decompose}, $\mathcal{T}$ is a solution to $\mathcal{I}$.
\end{proof}

The following corollary of Theorem~\ref{thm: all pe are half}, extends the result from the case where all $p_e=1/2$ to the case where all $p_e$ are equal.

\begin{corollary} \label{cor: all pe are equal}
ExecuteTutteStrategy($\pi$) is an optimal algorithm for the adaptive MST-ID problem, for every permutation $\pi$ of the edges of $G$, in the special case where all $p_e$ are equal. 
\end{corollary}

\begin{proof}
As in the proof of Theorem~\ref{thm: all pe are half},
consider the root-leaf path in $\mathcal{T}$ that would be followed on some
realization $\alpha$ of the weights of the edges of $G$.  Again, let $m$ be the number of edges of $G$. Let $n$ be the number of vertices.  
Let $\hat{p}$ be the common value of all the $p_e$ and let $q = 1-\hat{p}$.
Let $x^iy^j$ be the term corresponding to the leaf at the end of the path.  
The total number of queries along the path is $m-(i+j)$.
Since the output MST has $n-1$ edges,
the number of contractions along the path is $n-1-i$ and the number of deletions is $m-(n-1)-j$.  Therefore, the probability that the algorithm takes this path is $q^{n-1-i}(1-q)^{m-(n-1)-j}$.  
 Thus the expected cost of the strategy is equal to $\sum_{i,j} t_{i,j}(m-(i+j))q^{n-1-i}(1-q)^{m-(n-1)-j}$.  Since the values of the $t_{i,j}$ do not depend on the ordering of the edges in the tree, the corollary follows by the same argument used in the proof of Theorem~\ref{thm: all pe are half}.
\end{proof}

The last proof of this section concerns the MST-WD problem.

\begin{proof}[Proof of Theorem \ref{thm: mst wd optimal}]
A partial realization $\beta$ determines the MST weight iff every unqueried edge has both endpoints in the same component of $G^0_\beta$. For necessity, suppose an unqueried edge $e$ connects two distinct components. Assigning weight $1$ to all unqueried edges, and then changing only the weight of $e$ to $0$, yields two extensions with different MST weights. Sufficiency follows because every extension has the same weight 0 components, and the MST weight is one less than their number. During MST-WD ALG, the components of $G^0_\beta$ only merge, so every skipped edge remains within a component. At termination, $T_0$ spans these components and the queried $1$-weight edges connect them; hence Step~3 returns an MST of known weight.

For analyzing the query complexity, consider the residual multigraph obtained by contracting the components of $G^0_\beta$ and retaining only unqueried edges. We first prove optimality of the ordering $\pi$ among all {\em non-adaptive strategies}, namely, strategies that process edges in a fixed order, skipping any loops encountered along the way. We establish an exchange argument for two orders differing only by swapping consecutive edges $e$ and $f$. Condition on the history before these edges are processed. By independence, the unqueried edge weights retain their original distributions. If either edge is a loop, the order does not affect the number of queries on the pair. If neither edge is a loop and they are not parallel, both are queried regardless of their order. If they are parallel, processing $e$ first uses $1+p_e$ queries on the pair in expectation, whereas processing $f$ first uses $1+p_f$ queries in expectation. Moreover, for every realization, the graphs resulting from the two orders are identical after loops are discarded. Thus, putting $e$ before $f$ cannot increase the expected cost whenever $p_e\leq p_f$. By repeated adjacent exchanges, processing edges in nondecreasing order of $p_e$ is optimal among these fixed-order strategies.

We now prove, by induction on the number of remaining non-loop edges, that the greedy strategy is optimal among all {\em adaptive strategies} as well from every residual state. The claim is immediate when only loops remain. Fix the first query $e$ of an optimal strategy. By the induction hypothesis, its continuation after either outcome can be replaced by the greedy strategy on the resulting graph, with ties broken according to $\pi$, without increasing the expected cost. The resulting strategy processes $e$ first and then processes the remaining edges in the order induced by $\pi$, skipping loops. By iteratively applying the exchange argument above, sorting the entire order cannot increase the expected cost. Therefore, MST-WD ALG is optimal among all adaptive strategies.
\end{proof}

\section{Amortization lemma} \label{app: amortization}

\begin{lemma}
\label{lem:amortize}
$DecP$ achieves minimum expected amortized cost for the special case of the MST-ID problem where $p_e \leq \frac{1}{2}$ for all $e \in E$.
\end{lemma}

\begin{proof}
    Recall that under the amortized charging scheme, querying some edge $e$ incurs an expected cost of $2(1-p_e) = 2 \cdot P[w_e = 0]$. We denote $C(S)$ to be the total expected amortized cost of some strategy $S$. 

    Take some weighted graph $G = (V, E)$. We prove by induction on $m$ that $DecP$ minimizes the expected amortized cost among all skipping strategies when $p_e \leq \frac{1}{2}$ for all $e \in E$.

    \textbf{Base case: } $m=0$ so trivially $C(DecP)=0$.

    \textbf{Inductive hypothesis (IH): } Assume $DecP$ is the strategy that achieves minimum expected amortized cost for all graphs with fewer than $m$ edges. 
    
    Let $a = \arg\max_{e} p_e$, and let $Q$ be an optimal strategy on $G$. Observe the following cases:

    \textbf{Case 1 ($Q$ queries $a$ first):} By IH, $DecP$ is optimal on both subtrees $G\setminus a$ and $G/a$. Replacing $Q's$ subtrees with $DecP$, gives $DecP$ on $G$, so $C(DecP) = c(Q) = OPT$.

    \textbf{Case 2 ($Q$ queries $c \neq a$ first): } By IH, replacing $Q$'s subtrees after the first query with $DecP$ does not increase expected cost. Denote this modified strategy as $Q'$. Given that $a = \arg\max_{e} p_e$, $a$ remains the maximum element in both $G \setminus c$ and $G /c$. Thus, we have that $Q'$ queries $c$ first and then queries $a$ on both branches (unless $a$ is ``skipped'' via becoming a bridge or a loop). We aim to show that $C(DecP) \leq C(Q')$. There are four structural cases.

    \paragraph{Case A: Both branches of $Q'$ certify an MST after querying $c$.} $\{a, c\}$ must be a cycle of size 2 (a parallel pair), and all other elements must be loops or bridges. Parallel edges have symmetric structural behavior regardless of whether $a$ or $c$ is queried first. Thus, we have that $DecP$ also certifies an MST after one query. Given this, we have that $C(Q') = 2(1-p_c)$ and $C(DecP) = 2(1-p_a)$ giving $C(DecP) - C(Q') = 2(p_c - p_a) \leq 0$.

    \paragraph{Case B: Both branches of $Q'$ query $a$ after $c$.} In this case $a$ and $c$ are not parallel, and there exists some cycle through $a$ that does not contain $c$. By cycle elimination, it is easy to show that the symmetric argument holds, and there exists some cycle through $c$ that does not contain $a$. Given this, we have that $DecP$ also queries both elements and both branches of its decision tree. Since both strategies query both elements, the resulting minor is identical and we have that $C(Q') = C(DecP)$.

    \paragraph{Case C: For $Q'$, if $w_c = 0$ query $a$, if $w_c = 1$ MST is certified}

    In this case, $\{a, c\}$ is not a cycle, and every cycle containing $a$ must also contain $c$. Once again, through cycle elimination it's easy to show that the symmetric argument holds, and every cycle containing $c$ must also contain $a$. Given this, we have that $DecP$ mirrors $Q'$ on both of its branches. Thus, we have $C(DecP) - C(Q') = 2(p_c - p_a) \leq 0$.
    
    \paragraph{Case D: For $Q'$, if $w_c = 0$ MST is certified, if $w_c = 1$ query $a$}

    In this case, $\{a, c\}$ are together in a multi-edge of size $\geq 3$. Edges in a multi-edge have symmetric structural behavior, so we have that $DecP$ mirrors $Q'$. Thus, we have that $C(Q') = C(DecP) = 2(1-p_ap_c)$.

    \bigskip

    In all cases, $C(DecP) \leq C(Q') \leq C(Q) = OPT$, which completes our proof.
    
\end{proof}

\section{Proof of \#P-hardness for adaptive expected cost} \label{app: hardness}

\begin{proof}[Proof of Lemma~\ref{lemma: expected cost of tutte}]
    The strategy corresponds to a decision tree that is derived from the recursive definition of the Tutte polynomial.  Each leaf of this tree corresponds to a spanning tree of $G$ that is reached by the path from the root down to this leaf.  The leaf corresponds to the minor of $G$ produced by performing the contractions and deletions on the path.   Let $i$ denote the number of bridges of the minor and $j$ denote the number of loops in the minor.  The number of contractions along the path is $n-1-i$ and the number of deletions is $m-(n-1)-j$.
    Therefore, the probability that the algorithm takes this path and ends up in the leaf we are considering is $q^{n-1-i}(1-q)^{m-(n-1)-j}$.  
    The total number of contractions and deletions along the path is $m-(i+j)$.  

    The number of leaves whose minor has exactly $i$ bridges and $j$ loops, is, by the recursive definition of the Tutte polynomial, exactly the coefficient $t_{i, j}$. Summing the contribution $q^{n-1-i}(1-q)^{m-(n-1)-j}(m-(i+j))$ of each leaf over all possible $(i, j)$, weighted by $t_{i,j}$ gives the expression for $C(q)$. To establish the second claim, note that the derivative of $T_D(x)$ is:  $T'_D(x)= \sum_{i,j} t_{i,j} x^{i+j-1}(i+j)$. Substituting $q = \frac{1}{2}$ into the expression stated in Lemma \ref{lemma: expected cost of tutte} gives

$$C(\tfrac{1}{2}) = (\tfrac{1}{2})^m\sum_{i,j} t_{i,j} 2^{i+j}(m-(i+j))=(\tfrac{1}{2})^m(mT_D(2)-2T'_D(2))$$

It is immediate from the rank-nullity definition of the Tutte polynomial that $T_D(2) = T(2,2) = \sum_{A \subseteq E} 1^{r(E)-r(A)+\nu(A)} = 2^m$ (since it counts the total number of subgraphs of $G$). Therefore,
$$C(\tfrac{1}{2}) = (\tfrac{1}{2})^m\sum_{i,j} t_{i,j} 2^{i+j}(m-(i+j))=m-(\tfrac{1}{2})^{m-1}T'_D(2)$$

Thus we can compute $C(\tfrac{1}{2})$ in polynomial time iff we can compute $T_D'(2)$ in polynomial time.
\end{proof}

We now present the proof of Theorem \ref{thm: sharp P hard}, which is based on the two lemmas presented below. Throughout, we assume $G = (V, E)$ is a connected multigraph where all edges have the same probability $q$ of having weight 0. Let $n=|V|$ and $m=|E|$. To establish $\#P$-hardness for the general problem, it suffices to prove hardness for the special case $q = \frac{1}{2}$.

For $A \subseteq E$, let $G_A=(V,A)$.  
 Let $\kappa(A)$ be the number of connected components of $G_A$, $r(A) = n-\kappa(A)$ be the rank of $G_A$, and $\nu(A) = |A| - r(A)$ be the nullity of $G_A$. Let
$S(G) = \sum_{A \subseteq E} \kappa(A)$. 

\begin{lemma} \label{lem: computable if S(G) is}
$T'_D(2)$ is computable in polynomial time iff $S(G)$ is.
\end{lemma}

\begin{proof}
    Consider the alternative (rank nullity) formulation of $T_D(x)=T(x,x)$:

$$T_D(x) = \sum_{A \subseteq E}(x-1)^{r(E)-r(A)}(x-1)^{\nu(A)}=\sum_{A \subseteq E}(x-1)^{r(E)-r(A)+\nu(A)}$$

We assumed $G$ was connected, so $r(E) = n-1$ and therefore
$r(E)-r(A)=\kappa(A)-1$.
Clearly $\nu(A) = |A| - (n-\kappa(A))$.  

\begin{align}
T_D'(2) &= \sum_{A \subseteq E} (2\kappa(A)-
1+|A|-n)\\&=(2\sum_{A \subseteq E} \kappa(A))-2^{m}(1+n)+\sum_{A \subseteq E} |A|
\\&=(2\sum_{A \subseteq E} \kappa(A))-2^{m}(1+n)+m2^{m-1}
\\&=(2S(G))-2^{m}(1+n)+m2^{m-1}
\end{align}
where the last equality comes from using $\sum_{A\subseteq E} |A| = m2^{m-1}$, since each of the $m$ items appears in half of the subsets.

Thus $T_D'(2)$ is computable in polynomial time iff $S(G)$ is.  
\end{proof}

\begin{lemma}
\label{lem:computing S(G) is hard}
    $S(G)$ is $\#$P-hard to compute.
\end{lemma}

\begin{proof}
    Recall that $S(G)=\sum_{A \subseteq E} \kappa(A)$.  We show that $S(G)$ is \#P-hard to compute, via a reduction from two terminal network reliability. 
    
For vertices $s \neq t$ in $V$, let
$D_{s,t}(G) = |\{A \subseteq E : s \not\leftrightarrow t \text{ in } G_A\}|$, where $s \not\leftrightarrow t$ in $G_A$ means there is no path from $s$ to $t$ in $G_A$.   
    The two terminal network reliability problem takes as input a connected graph $G=(V,E)$ and
    vertices $s \neq t$ in $G$.  The problem is to compute $\mathrm{Rel}_{s,t}(G) = |\{A \subseteq E : s \leftrightarrow t \text{ in } G_A\}|$, the number of subsets $A \subseteq E$ for which there is a path from $s$ to $t$ in $G_A$.

    Let $G^+ := G \cup \{e\}$ where $e = (s,t)$ is a new edge not already in $G$. We can partition the subsets of $E(G^+)$ according to whether they contain $e$, and then analyze how much each side of the partition contributes to the value of $S(G^+)$.

\begin{enumerate}
    \item Subsets not containing $e$ contribute $S(G)$, since connectivity for these subsets is the same as in $G$
    \item Subsets containing $e$ contribute $S(G) - D_{s,t}(G)$, since adding $e$ decreases $\kappa(A)$ by exactly $1$ when $s$ and $t$ are disconnected in $G_A$
\end{enumerate}

Summing both of the partition contributions together gives 
$$S(G^+) = 2S(G) - D_{s,t}(G),$$
and rearranging gives
$$D_{s,t}(G) = 2S(G) - S(G^+).$$

Thus an oracle for $S$ yields $D_{s,t}(G)$ via two calls, and hence also 

$$\mathrm{Rel}_{s,t}(G) = 2^{m} - D_{s, t}(G)$$

Since two terminal reliability is \#P-hard (\cite{valiant1979complexity}, Theorem 1), $S(G)$ is \#P-hard as well.
\end{proof}

\begin{proof}[Proof of Theorem~\ref{thm: sharp P hard}]

The proof follows immediately from Lemmas~\ref{lemma: expected cost of tutte}, \ref{lem: computable if S(G) is}, and \ref{lem:computing S(G) is hard}.
\end{proof}

\section{Deferred proofs for non-adaptive MST-ID} \label{app: non-adaptive}

\begin{proof} [Proof of Lemma~\ref{lem:nacert}]
Let $\beta := \alpha^S$. By Lemma \ref{prop: beta certifies if no loop or bridge}, $\beta$ is an MST-ID certificate for $G$ iff the minor $G_\beta$ contains only loops and bridges.

\paragraph{Necessity:}
Suppose $S$ does not contain a set $S_1$ satisfying condition (1). Then the unqueried weight 1 edges ($E^* \setminus S$) contain a cycle in $G^*$. If $E^* \setminus S$ were acyclic, then it could be extended to a spanning tree $E^* \setminus S_1$ of $G^*$, giving $S_1 \subseteq S$, which is a contradiction. Since $G^*$ is formed by contracting weight 0 edges and deleting loops, these unqueried weight 1 edges remain as a non-loop cycle in $G_\beta$. By Lemma \ref{prop: beta certifies if no loop or bridge}, this contradicts that $\beta$ is a certificate for an MST of $G$.

Suppose $S$ does not contain a set $S_2$ satisfying the conditions in (2) above.  Then there is some $c' \in C'$ such that the edges contracted in forming $G_{\beta}$ from $G$ do not span the vertices of $c'$.  Since $c'$ is a 2-edge-connected component of weight 0 edges, it follows that the contractions performed in forming $G_{\beta}$ will not merge all vertices of $c'$ into a single vertex and $G_{\beta}$ will still have a non-loop cycle.  Contradiction.  

Finally, suppose $S$ does not contain a set $S_3$ satisfying the conditions in (3). Then there is a $c^* \in C^*$ such that the edges in $S$ that correspond to vertices in $B_{c^*}$ do not cover the edges of $B_{c^*}$. Consider the endpoints of an uncovered edge in $B_{c^*}$: the endpoint in $R$ corresponds to a weight 1 edge $e_R$ in $G'$ connecting two different vertices of $t_{c^*}$, and the endpoint in $L$ consists of an edge $e_L$ in $t_{c^*}$ that is in the fundamental cycle that $e_R$ forms with $t_{c^*}$.  
We will show that the fact that $e_R$ and $e_L$ are not in $S$ leads to a contradiction.

Because each edge in $t_{c^*}$ is 
a bridge edge in the subgraph of $G'$ consisting only of its weight 0 edges, each edge of $t_{c^*}$ is a unique minimum-weight edge in a cut of $G'$, and hence also in a cut of $G$.   
Thus any MST $T$ of $G$ must contain all edges of $t_{c^*}$, including edge $e_L$. Further, $T$ cannot contain any weight 1 edges with both endpoints in $c^*$, since the vertices of $c^*$ can be spanned by weight 0 edges. 

Therefore, $e_R$ is not in $T$. Since $e_L$ lies on the fundamental cycle formed by $e_R$ with $t_{c^*}$ in $G'$, and $T$ spans $c^*$ using exclusively weight 0 edges, $e_L$ must also lie on the fundamental cycle that $e_R$ forms with $T$ in $G$. Thus $S$ does not contain the weight 1 edge $e_R$, nor does it contain all the edges in the fundamental cycle $e_R$ forms with $T$. By Lemma~\ref{prop: beta certifies if no loop or bridge}, $\beta$ does not certify $T$ as an MST of $G$, contradiction.

\bigskip
{\bf \noindent Sufficiency:}
Suppose $S_1 \cup S_2 \cup S_3 \subseteq S$, with $S_1$,$S_2$,$S_3$ having the properties given in (1),(2),(3).

Let $D$ be a cycle in $G$ that is not a loop. 
Forming $G_{\beta}$ from $G$ involves performing contractions and deletions.  
We will show that these operations destroy cycle $D$ (either by deleting an edge in $D$ or by merging all vertices of $D$ into a single vertex).

Since the edges in $S_1$ have weight 1, the deletions include deleting all edges in $S_1$.  
If $D$ contained edges with endpoints in more than one zero-connected component $c^* \in C^*$, that edge would be deleted when deleting the edges in $S_1$, destroying the cycle.  
Suppose $D$ does not contain any such edges.  Then deleting the edges in $S_1$ does not affect $D$.  Consider the effect on $D$ of contracting the edges in $S_2$. If cycle $D$ has no edges with endpoints in different components $c' \in C'$, then the contractions cause the vertices of $D$ to be merged into a single vertex, again destroying cycle $D$.

Thus if after deleting the edges in $S_1$ and contracting the edges in $S_2$, $D$ has not been destroyed, the remaining edges of $D$ form a non-empty set we'll call $K$, consisting of edges with endpoints in different $c' \in C'$, but in the same $c^* \in C^*$. 
 Let $K_1$ be the set of weight 1 edges of $K$; it is non-empty because the weight 0 edges of $K$ are all contained in the (acyclic) tree $t_{c^*}$.  Consider an edge 
$\hat{e} \in K_1$.  It corresponds to a vertex of $B_{c^*}$. 
Now suppose we contract or delete the edges in $S_3$, according to their weights.  
This either causes contraction of a path of edges in $t_{c^*}$ between the endpoints of $\hat{e}$, or deletion of $\hat{e}$.  If at least one $\hat{e} \in K_1$ is deleted, then $D$ is destroyed.  Suppose that is not the case.  Then
the endpoints of each edge $\hat{e} \in K_1$ have been merged into a common vertex, causing $\hat{e}$ to become a loop.
Consider the remaining set of edges in $D$.
They are all in $K \setminus K_1$, so they are weight 0 edges, and are in $t_{c^*}$.  Since $t_{c^*}$ is acyclic, 
these edges cannot form a cycle.  Therefore cycle $D$ has been destroyed.
\end{proof}

Below, we present the proof of Lemma \ref{prop:certificate if goal value}.

\begin{proof}

By the definition of $g_\alpha$,
\[
g_\alpha(S) = Q_\alpha \iff g^*(S)=Q^* \ \text{ and } \ \forall c'~(g_{c'}(S)=Q_{c'}) \  \ \text{ and } \
\forall {c^*}~(g_{c^*}(S)=Q_{c^*}) \ .
\]
It remains to show this condition holds iff $\alpha^S$ is a certificate of an MST of $G$.
By Lemma~\ref{lem:nacert}, $\alpha^S$ is a certificate of an MST of $G$ iff it contains subsets $S_1,S_2,S_3$ satisfying the properties stated in the lemma.   It follows easily from the definitions of $g^*$, the $g_{c'}$, and the $g_{c^*}$ that the condition given for $\alpha^S$ in Lemma~\ref{lem:nacert} is satisfied iff the above condition is satisfied. 

It remains to show that $Q_\alpha = O(E^2)$. Note that $Q^* \leq |E|$. The sum $\sum_{c'} Q_{c'} \leq |V| - 1$
since the $c'$ do not share vertices.
For $\sum_{c^*} Q_{c^*} = \sum_{c^*} |E(B_{c^*})|$, each edge of $B_{c^*}$ corresponds to a
(weight 0 edge, weight 1 edge) pair where the weight 0 edge lies on the fundamental cycle
the weight 1 edge forms with tree $t_{c^*}$.
Since together the trees $t_{c^*}$ have less than $|V|$ edges total, 
 $\sum_{c^*} Q_{c^*} \leq |E||V| \leq |E|^2$ since $G$ is connected.
Hence $Q_\alpha=O(E^2)$. Therefore, we have proved Lemma \ref{prop:certificate if goal value}.
\end{proof}

\providecommand{\figwidth}{1\textwidth}
\begin{figure}[htbp]
\centering
\resizebox{\figwidth}{!}{%
\begin{tikzpicture}[
  v/.style={circle,fill=black,inner sep=1.2pt},
  big/.style={circle,draw=black,fill=white,inner sep=1.6pt,font=\scriptsize},
  ed/.style={thick},
  lab/.style={font=\scriptsize,inner sep=1pt,fill=white,fill opacity=.9,text opacity=1},
  ccbox/.style={draw=black!40,rounded corners=7pt,fill=black!7,inner sep=6pt},
  subbox/.style={draw=blue!60,dashed,thick,rounded corners=5pt,fill=blue!5,inner sep=3pt},
  tbox/.style={draw=black!70,dotted,thick,rounded corners=7pt,inner sep=6pt},
  ttl/.style={font=\small\bfseries,anchor=north west},
  flow/.style={-{Stealth[length=3mm]},very thick,black!70},
  flab/.style={font=\scriptsize,align=center,text width=2.4cm}
]

\begin{scope}[shift={(0,0)}]
  \node[ttl] at (-2.1,4.9) {(a) $G$};
  \coordinate (aL) at (-1.75,0);   \coordinate (aT) at (0.45,1.5);
  \coordinate (aB) at (0.45,-1.5); \coordinate (aR) at (2.5,0);
  \coordinate (bT) at (1.2,3.85);  \coordinate (bB) at (1.2,2.75);
  \coordinate (cR) at (5.05,0);
  \node[v] (v1) at (-1.1, 0.75) {};  \node[v] (v2) at (-1.1,-0.75) {};
  \node[v] (v3) at ( 0.7, 0)    {};
  \node[v] (v4) at ( 2.0, 1.0)  {};  \node[v] (v5) at ( 2.0,-1.0)  {};
  \node[v] (u1) at ( 0.2, 3.3)  {};  \node[v] (u2) at ( 2.2, 3.3)  {};
  \node[v] (w1) at ( 4.4, 1.0)  {};  \node[v] (w2) at ( 4.4,-1.0)  {};
  \begin{pgfonlayer}{background}
    \node[ccbox,fit=(v1)(v2)(v3)(v4)(v5)(aL)(aT)(aB)(aR),label={[font=\scriptsize]above left:$c^*_A$}] {};
    \node[ccbox,fit=(u1)(u2)(bT)(bB),label={[font=\scriptsize]above:$c^*_B$}] {};
    \node[ccbox,fit=(w1)(w2)(cR),label={[font=\scriptsize]above right:$c^*_C$}] {};
    \node[subbox,fit=(v1)(v2),label={[font=\scriptsize,blue!80!black]left:$A$}] {};
    \node[subbox,fit=(u1)(u2),label={[font=\scriptsize,blue!80!black]left:$B$}] {};
  \end{pgfonlayer}
  \draw[ed] (v1) to[bend left=40]  node[lab,pos=.3]{$a_1,0$} (v2);
  \draw[ed] (v1) to[bend right=40] node[lab,pos=.7]{$a_2,0$} (v2);
  \draw[ed] (v2) -- node[lab,midway]{$a_3,0$} (v3);
  \draw[ed] (v3) -- node[lab,midway]{$a_4,0$} (v4);
  \draw[ed] (v3) -- node[lab,midway]{$a_5,0$} (v5);
  \draw[ed] (v4) to[bend left=25]  node[lab,pos=.5]{$x_5,1$} (v5);
  \draw[ed] (v4) to[bend right=22] node[lab,pos=.5]{$x_6,1$} (v1);
  \draw[ed] (v5) to[bend left=22]  node[lab,pos=.5]{$x_7,1$} (v2);
  \draw[ed] (u1) to[bend left=30]  node[lab,pos=.5]{$b_1,0$} (u2);
  \draw[ed] (u1) to[bend right=30] node[lab,pos=.5]{$b_2,0$} (u2);
  \draw[ed] (w1) -- node[lab,pos=.27]{$c_1,0$} (w2);
  \draw[ed] (w1) to[bend left=50] node[lab,pos=.73]{$x_8,1$} (w2);
  \draw[ed] (v1) -- node[lab,midway]{$x_1,1$} (u1);
  \draw[ed] (v4) -- node[lab,midway]{$x_2,1$} (u2);
  \draw[ed] (u2) -- node[lab,midway]{$x_3,1$} (w1);
  \draw[ed] (v5) -- node[lab,midway]{$x_4,1$} (w2);
\end{scope}

\begin{scope}[shift={(9.4,0.2)}]
  \node[ttl] at (-0.9,4.7) {(b) $G^{*}$};
  \node[big] (As) at (0,0.2)   {$c^*_A$};
  \node[big] (Bs) at (1.9,3.0) {$c^*_B$};
  \node[big] (Cs) at (3.8,0.2) {$c^*_C$};
  \draw[ed] (As) to[bend left=16]  node[lab]{$x_1,1$} (Bs);
  \draw[ed] (As) to[bend right=16] node[lab]{$x_2,1$} (Bs);
  \draw[ed] (Bs) -- node[lab]{$x_3,1$} (Cs);
  \draw[ed] (As) -- node[lab]{$x_4,1$} (Cs);
\end{scope}

\begin{scope}[shift={(0,-9.2)}]
  \node[ttl] at (-2.1,4.9) {(c) $G'$};
  \coordinate (tL) at (-1.75,0); \coordinate (tT) at (0.45,1.2);
  \coordinate (tB) at (0.45,-1.2); \coordinate (tR) at (2.5,0);
  \coordinate (cR2) at (5.05,0);
  \node[big] (A1) at (-1.1,0) {$A$};
  \node[v]   (v3) at ( 0.7,0)    {};
  \node[v]   (v4) at ( 2.0,1.0)  {};  \node[v] (v5) at (2.0,-1.0) {};
  \node[big] (B)  at ( 1.2,3.3) {$B$};
  \node[v]   (w1) at ( 4.4,1.0)  {};  \node[v] (w2) at (4.4,-1.0) {};
  \begin{pgfonlayer}{background}
    \node[tbox,fit=(A1)(v3)(v4)(v5)(tL)(tT)(tB)(tR),label={[font=\scriptsize]below left:$c^*_A$}] {};
    \node[tbox,fit=(B),label={[font=\scriptsize]above right:$c^*_B$}] {};
    \node[tbox,fit=(w1)(w2)(cR2),label={[font=\scriptsize]above right:$c^*_C$}] {};
  \end{pgfonlayer}
  \draw[ed] (A1) -- node[lab,midway]{$a_3,0$} (v3);
  \draw[ed] (v3) -- node[lab,midway]{$a_4,0$} (v4);
  \draw[ed] (v3) -- node[lab,midway]{$a_5,0$} (v5);
  \draw[ed] (v4) to[bend left=25]  node[lab,pos=.5]{$x_5,1$} (v5);
  \draw[ed] (v4) to[bend right=22] node[lab,pos=.5]{$x_6,1$} (A1);
  \draw[ed] (v5) to[bend left=22]  node[lab,pos=.5]{$x_7,1$} (A1);
  \draw[ed] (w1) -- node[lab,pos=.27]{$c_1,0$} (w2);
  \draw[ed] (w1) to[bend left=50] node[lab,pos=.73]{$x_8,1$} (w2);
  \draw[ed] (A1) -- node[lab,midway]{$x_1,1$} (B);
  \draw[ed] (v4) -- node[lab,midway]{$x_2,1$} (B);
  \draw[ed] (B)  -- node[lab,midway]{$x_3,1$} (w1);
  \draw[ed] (v5) -- node[lab,midway]{$x_4,1$} (w2);
\end{scope}

\begin{scope}[shift={(9.4,-9.2)}]
  \node[ttl] at (-0.9,4.9) {(d) Bipartite graphs $B_{c^*}$};
  \tikzset{L/.style={draw=black,thick,rectangle,fill=white,inner sep=2.5pt,font=\scriptsize},
           R/.style={draw=black,thick,circle,fill=white,inner sep=1.5pt,font=\scriptsize}}
  \node[font=\scriptsize,align=center] at (0.2,3.85) {left side:\\bridge edges};
  \node[font=\scriptsize,align=center] at (3.6,3.85) {right side:\\internal weight 1 edges};
  \node[L] (la3) at (0.2,2.7) {$a_3$};
  \node[L] (la4) at (0.2,1.6) {$a_4$};
  \node[L] (la5) at (0.2,0.5) {$a_5$};
  \node[R] (rx6) at (3.6,2.7) {$x_6$};
  \node[R] (rx5) at (3.6,1.6) {$x_5$};
  \node[R] (rx7) at (3.6,0.5) {$x_7$};
  \draw[ed] (la3)--(rx6);
  \draw[ed] (la4)--(rx6);
  \draw[ed] (la4)--(rx5);
  \draw[ed] (la5)--(rx5);
  \draw[ed] (la5)--(rx7);
  \draw[ed] (la3)--(rx7);
  \node[font=\scriptsize,anchor=west] at (4.5,1.6) {$B_{c^*_A}$};
  \node[L] (lc1) at (0.2,-0.9) {$c_1$};
  \node[R] (rx8) at (3.6,-0.9) {$x_8$};
  \draw[ed] (lc1)--(rx8);
  \node[font=\scriptsize,anchor=west] at (4.5,-0.9) {$B_{c^*_C}$};
\end{scope}

\draw[flow] (5.9,2.0) -- (8.5,2.0);
\node[flab,above] at (7.2,2.15) {contract each weight 0 component};
\draw[flow] (1.0,-2.2) -- (1.0,-4.2);
\node[flab,right] at (1.15,-3.2) {contract each 2-edge-connected component $c'$};
\draw[flow] (5.9,-7.2) -- (8.5,-7.2);
\node[flab,above] at (7.2,-7.05) {components $c^*$ and their graphs $B_{c^*}$};
\end{tikzpicture}%
}
\caption{Worked example of the goal function construction given some realization $\alpha$ on a graph $G$. Each edge is labeled as (edge name, edge weight). \textbf{(a)} The original graph $G$, showing weight-0 components $c^*$ (outer shaded boxes) and 2-edge-connected weight-0 components $c'$ (inner dashed boxes $A, B$). \textbf{(b)} $G^*$ is obtained by contracting each weight-0 connected component $c^*$. \textbf{(c)} $G'$ is obtained by contracting each 2-edge-connected component $c'$. What remains inside each $c^* \in C^*$ is a tree of its bridge edges ($t_{c^*}$). \textbf{(d)} The bipartite graphs $B_{c^*_A}$ and $B_{c^*_C}$ constructed from $t_{c^*}$ and internal weight 1 edges.}
\label{fig:goalvalue-example}
\end{figure}
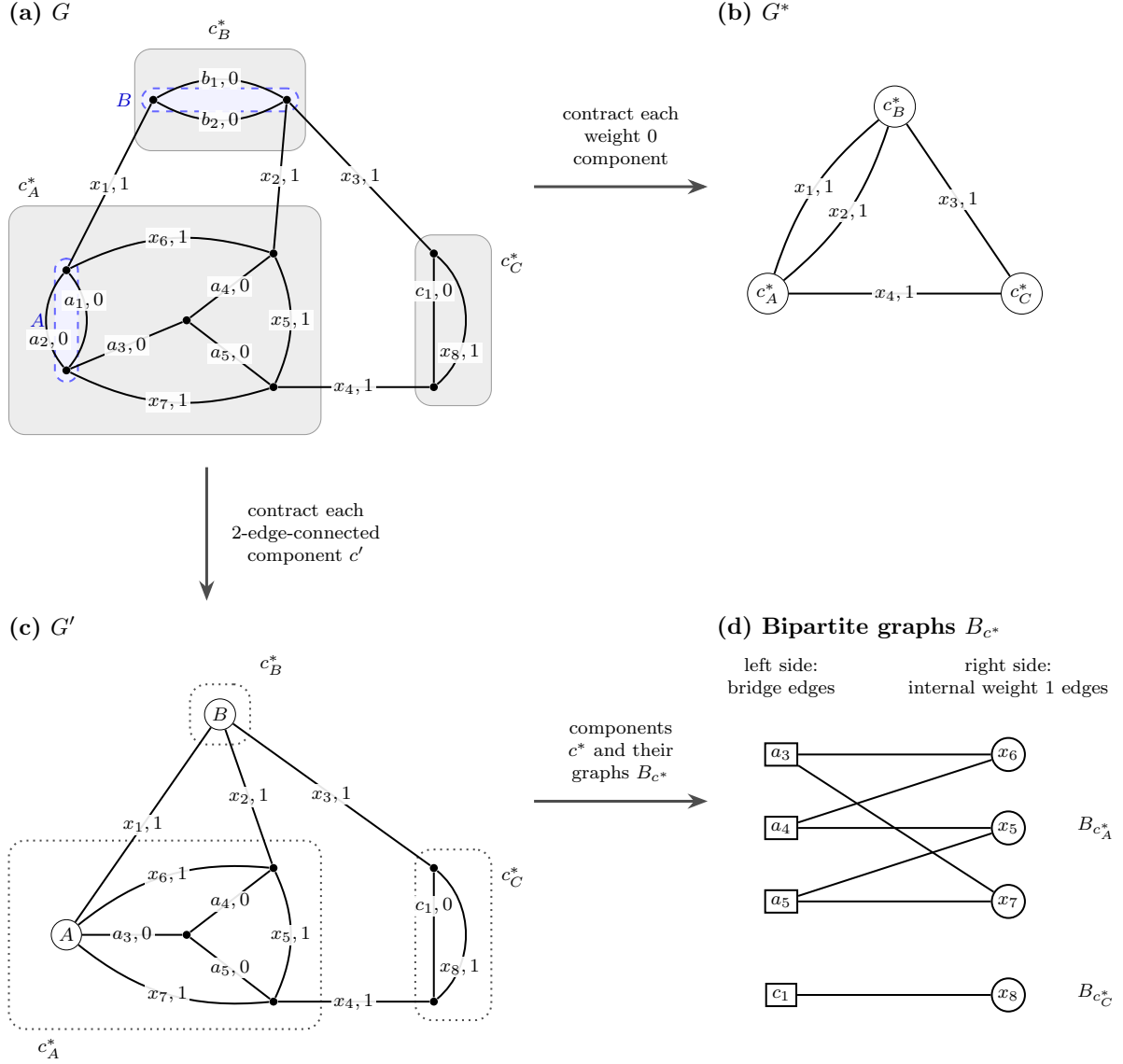

\section{Bounds on submodular goal values} \label{app: goal values}

\subsection{Lower bound lemma for submodular goal values} 

The following lemma can be useful in proving lower bounds on submodular goal value for some graphs.

\begin{lemma}
\label{lem:forlbQ}
Let $G=(V,E)$ be a graph, $e \in E$, and $g:\{0,1,*\}^E \rightarrow \{0,1,2,\ldots,Q\}$ be a submodular goal function for the MST-ID problem on $G$. 
Let $\beta \in \{0,1,*\}^E$ such that $\beta(e)=*$.  Then
\begin{enumerate}
    \item[\normalfont(i)] if $e$ lies on a simple cycle of length $\ell$ in the minor $G_\beta$, then 
     $ g(\beta_{e \leftarrow 1}) - g(\beta) \ge \ell - 1$
    \item[\normalfont(ii)] dually, if $e$ is a member of a minimal cut-set (bond) of size $\ell$ in the minor $G_\beta$, then 
      $g(\beta_{e \leftarrow 0}) - g(\beta) \ge \ell - 1$
\end{enumerate}
\end{lemma}

\begin{proof}
    Let $g$ be a submodular goal function for $G$ with goal value $Q$. 

    \paragraph{Part (i):} Let $C$ be the simple cycle of length $\ell$ containing edge $e$ in the minor $G_\beta$. Let $\gamma \in \{0,1,*\}^E$ be the extension of $\beta$ produced by deleting all edges $e'$ in the minor $G_\beta$ except for those in $C$. Then the resulting minor $G_\gamma$ consists solely of the cycle $C$. Recall from earlier that identifying an MST on a simple cycle requires querying until either a weight 1 edge is identified or $\ell - 1$ weight 0 edges are found. In the worst case, resolving $C$ requires $\ell-1$ queries. Given that $g$ increases by at most $1$ per query, the goal value distance satisfies $$Q - g(\gamma) \ge \ell - 1$$

    Setting $e$ to $1$ deletes $e$ from $C$, leaving $G_{\gamma_{ e \leftarrow 1}}$ as an acyclic path of length $\ell-1$. Since all edges in $G_{\gamma_{ e \leftarrow 1}}$ are bridges, we have that $\gamma_{e \leftarrow 1}$ is a certificate of some MST on $G$, and $g(\gamma_{e \leftarrow 1}) = Q$. Substituting gives $$g(\gamma_{e \leftarrow 1}) - g(\gamma) = Q - g(\gamma) \ge \ell - 1$$ Recall that $\beta \preceq \gamma$, so by the definition of submodularity we have $$g(\beta_{e \leftarrow 1}) - g(\beta) \ge g(\gamma_{e \leftarrow 1}) - g(\gamma) \ge \ell - 1$$

    \paragraph{Part (ii):} A minimum edge-cut (bond) of size $\ell$ containing $e \in G_\beta$ corresponds to a circuit of length $\ell$ in the cographic (dual) matroid $M^*(G_\beta)$ \cite{oxley}. Setting $e$ to $0$ contracts $e$, which merges the two components defined by the cut and resolves the dual circuit. Applying part (i) under graphic matroid duality gives $$g(\beta_{e \leftarrow 0}) - g(\beta) \ge \ell - 1$$
\end{proof}

\subsection{Submodular goal value of a multi-edge cycle}

\begin{theorem} \label{thm: multicycle and disjoint st path equal to formula}
Let $G$ be a multi-edge cycle.  Then the submodular goal value for MST-ID of $G$ is equal to the rank$\times$corank value for $G$.
Dually, the same holds if $G$ 
consists of disjoint paths between two vertices, $s$ and $t$.
\end{theorem}

\begin{proof}
Let $G$ be a multi-edge cycle.

\paragraph{Lower bound on submodular goal value:}
We begin by showing that the rank$\times$corank formula for $G$ is a lower bound for the goal value of any submodular goal function $g$ for MST identification of graph $G$.
Consider such a $g$.
 Let $m$ be the number of edges in the multi-edge cycle, and $n$ be the number of vertices (equivalently the number of multi-edges).  Choose 1 edge from each multi-edge to form a cycle $C$.  Define a sequence of partial assignments beginning with $(*,*,*,...,*)$.  Let $E'$ be the set of edges not in cycle $C$, and let $m'=|E'|$.  W.l.o.g. let $E'=\{e_1,e_2,\ldots,e_{m'}\}$.  The sequence is $\beta_0, \beta_1, \ldots, \beta_{m'}$ where $\beta_i$ is the assignment that sets $e_1,\ldots,e_{i}$ to 1 and all other edges to $*$.
So the sequence corresponds to deleting the edges in $E'$ one by one.  Consider a $\beta_i$ in the sequence, where $0 < i \leq m'$.  In the minor induced by $\beta_i$, the cycle $C$ is still present.  Also, $C$ contains an edge in the same multi-edge as $e_i$, that is, an edge that is parallel to $e_i$.  Replacing this edge by $e_i$ yields a cycle $C_i$ containing $e_i$.  Since this cycle is present in the minor corresponding to $\beta_i$, it follows from Lemma~\ref{lem:forlbQ} that $g(\beta_i) - g(\beta_{i-1}) \geq n-1$.  Therefore, $g(\beta_{m'}) \geq (n-1)m'$.  Finally, define $\beta_{m'+1}$ to be an assignment produced from $\beta_{m'}$ by deleting an arbitrary edge $e$ from $C$.  Since $\beta_{m'}$ contains $e$, it follows that $g(\beta_{m'+1})-g(\beta_{m'}) \geq (n-1)$ and hence  $Q \geq g(\beta_{m'+1})-g(\beta_{0}) \geq (n-1)(m'+1) = (n-1)(m-n+1)$, which is the rank$\times$corank formula for $G$. 

\paragraph{Construction of submodular goal function:}
We now construct a submodular goal function $g$ for MST identification for $G$ whose goal value $Q$ is equal to the rank$\times$corank formula.  We construct $g$ by applying the standard OR construction (cf.~\cite{bach2018submodular}) to two monotone submodular functions, $g_1$ and $g_2$, with goal values $Q_1=n-1$ and $Q_2 = m-n+1$.
The result of the OR construction is the function $g$ such that for $\beta \in \{0,1,*\}^E$, $g(\beta) = Q_1Q_2 -  (Q_1-g_1(\beta))(Q_2 - g_2(\beta))$.

The value of $g_1$ on any partial assignment $\beta$ is just the rank in the graphic matroid $M(G)$, of the subset of edges set to 0 by $\beta$.  Equivalently, it is the number of multi-edges in $G$ such that $\beta$ has set at least one edge in the multi-edge to 0.  
The value of $g_2$ on any assignment $\beta$ is equal to the sum of two functions, $g_2'$ and $g_2''$.
Let $\mathbb{M}$ be the set of all multi-edges of $G$.  For multi-edge $M\in\mathbb{M}$, define $N_1(M,\beta)$ to be the number of $e \in M$ set to 1 by $\beta$. 
Define $z(M,\beta)$ to be $|M|-1$ if $\beta$ sets at least one $e \in M$ to 0, else $z(M,\beta)=\min\{N_1(M, \beta), |M|-1\}$.   
The value of $g_2'(\beta)=\sum_{M \in \mathbb{M}} z(M,\beta)$.  The value of $g_2''(\beta)=1$ if $N_1(M, \beta)=|M|$ for some $M \in \mathbb{M}$, else $g_2''(\beta)=0$.  The maximum possible value of $g_2'(\beta)$ is $\sum_{M \in \mathbb{M}}(|M|-1) = m-n$, and it is reached if for each multi-edge $M$, either (1) $\beta$ assigns at least one edge in $M$ to 0 (contracting the edge, so that any edges $e \in M$ such that $\beta(e)=*$ are turned into self-loops), or (2) $\beta$ assigns 1 to at least $|M|-1$ of the edges in $M$. 
It therefore follows from the definition of $g_2''$ that $g_2(\beta)=m-n+1$ iff there is at least one multi-edge $M \in \mathbb{M}$ such that $\beta$ sets all $e \in M$ to 1 (deleting all edges in the multi-edge), and for all other multi-edges $M \in \mathbb{M}$, $\beta$ either sets at least one edge in $M$ to 0 (contracting it and turning $*$ edges in the multi-edge to self-loops) or sets at least $|M|-1$ edges in $M$ to 1 (deleting them, so that at most one edge remains in $M$).

\paragraph{Monotonicity and submodularity of $g$:}
Because it is a rank function, $g_1$ is clearly monotone and submodular. 
Function $g_2$ is clearly monotone.
To show $g_2$ is also submodular, consider $\beta, \gamma \in \{0,1,*\}^E$ such that $\beta \preceq \gamma$ and $e \in E$ with $\beta_e=\gamma_e = *$.  Note that $g_2'$ is both monotone and submodular, because for each $M \in \mathbb{M}$, the function $g_M := g(M,\cdot)$ is monotone and submodular.  We claim first that
$g_2(\beta_{e \leftarrow 0}) - g_2(\beta) \geq g_2(\gamma_{e \leftarrow 0}) - g_2(\gamma)$.  This holds because setting $e$ to 0 in $\beta$ only changes the value of $g_2'$, not the value of $g_2''$, and the same is true for setting $e$ to 0 in $\gamma$.  

To complete the proof of submodularity, it remains only to show that $g_2(\beta_{e \leftarrow 1}) - g_2(\beta) \geq g_2(\gamma_{e \leftarrow 1}) - g_2(\gamma)$.
Since $g_2'$ is submodular and the value of $g_2''$ is either 0 or 1, the only way this could be violated is if $g_2''(\beta_{e \leftarrow 1}) - g_2''(\beta) = 0$ and $g_2''(\gamma_{e \leftarrow 1}) - g''(\gamma)=1$.
So assume that is the case.
Clearly $g_2''(\gamma_{e \leftarrow 1}) - g''(\gamma) = 1$ implies that 
$e$ is a member of a multi-edge $M$ such that $\gamma$ sets all $|M|-1$ other edges of $M$ to 1.  Further, it implies that $g_2'(\gamma_{e \leftarrow 1}) - g'(\gamma) = 0$ and
for all multi-edges $M' \neq M$, $\gamma$ sets fewer than $|M'|$ of the edges in $M'$ to 1.  Since $\beta \preceq \gamma$, it also holds that for all multi-edges $M' \neq M$, $\beta$ sets fewer than $|M'|$ of the 
edges in $M'$ to 1.  Further, $\beta$ sets some $i \leq |M|-1$ of the edges in $M$ to 1, and the other edges in $M$ to $*$.  If $i = |M|-1$, then 
$g_2''(\beta_{e \leftarrow 1}) - g''(\beta) = 1$, which contradicts our assumption.
If $i < |M|-1$, then $g_2''(\beta_{e \leftarrow 1}) - g''(\beta) = 0$, but 
$g_2'(\beta_{e \leftarrow 1}) - g'(\beta) = 1$, so $g_2(\beta_{e \leftarrow 1}) - g_2(\beta) = 1$.
Since $g_2'(\gamma_{e \leftarrow 1}) - g_2'(\gamma)=0$, we have that 
$g_2(\gamma_{e \leftarrow 1}) - g_2(\beta) = 1$ and thus
$g_2(\beta_{e \leftarrow 1}) - g_2(\beta)= g_2(\beta_{e \leftarrow 1}) - g_2(\beta)=0$,
completing the proof that $g_2$ is submodular.

\paragraph{Certificate property of $g$: }
 The maximum values of $g_1$ and $g_2$ are $Q_1=n-1$ and $Q_2=m-n+1$ respectively.
Since $g$ is the result of applying the OR construction to these monotone submodular functions, $g$ is also monotone and submodular, and its goal value $Q=(n-1)(m-n+1)$ is attained
iff $g_1(\beta)=n-1$ or $g_2(\beta)=m-n+1$ (or both).
It remains to show that the information in $\beta$ is an MST-ID certificate for $G$ iff $g_1(\beta)=n-1$ or $g_2(\beta)=m-n+1$.

If $g_1(\beta)=n-1$, then the edges set to 0 by $\beta$ include a spanning tree of $G$ of total weight 0, which is clearly an MST, making $\beta$ an MST-ID certificate for $G$.  

If $g_2(\beta)=m-n+1$, then $g_2''(\beta)=1$, meaning that $\beta$ deletes all edges in at least one multi-edge $M^*$.  Further, $g_2'(\beta)=m-n$, meaning that in every multi-edge $M \neq M^*$, $\beta$ either contracts at least one edge (making all $*$'ed edges in the multi-edge into self-loops) or deletes all but one edge (making that possible remaining edge into a bridge, because $\beta$ deleted all edges in $M^*$).  Therefore, if $g_2(\beta)=m-n+1$, $G_\beta$ has only bridges and self-loops, so by Lemma \ref{prop: beta certifies if no loop or bridge}, $\beta$ is an MST-ID certificate for $G$.
 For the other direction, suppose $\beta$ is an MST-ID certificate for $G$. Then either (1) $\beta$ assigns a 0 to at least one edge in at least $n-1$ of the multi-edges, or (2) $\beta$ deletes all edges in one of the multi-edges, and for every other multi-edge, it must either set one of its edges to 0 (making all remaining $*$'s into self-loops) or set at least $|M|-1$ of its edges to 1 (making the one possible remaining $*$'ed edge, if it exists, into a bridge).   In the first case, $g_1(\beta) = n-1$ and in the second, $g_2(\beta) = m-n+1$. 

\paragraph{Dual statement:}
The dual statement then follows from the fact that multi-edge cycles are planar graphs whose dual graphs can be represented as disjoint paths of edges between two vertices $s$ and $t$, and vice-versa.  
\end{proof}

\subsection{Integer Linear Programming formulation and computational results}

Here, we describe how we computed the goal values for a number of small graphs. We report some of these values in Table \ref{tab: goal values}.

\paragraph{Graph types considered}
We computed submodular goal values for complete graphs $K_n$, fan graphs $F_k$ (a path of length $k$, where all vertices are connected to some hub vertex), wheel graphs $W_k$ (a simple cycle of length $k$ joined to a hub vertex, simple cycles $C_n$, complete bipartite graphs $K_{r, s}$, and theta graphs $\theta(a, b, c)$ (two pole vertices connected by three vertex disjoint paths of lengths $a, b, c$).

We also computed values for a few other graphs:
\begin{enumerate}
     \item \textbf{Diamond graph} (square $+$ diagonal): $\Theta(1,2,2)$, formed by a cycle of length 4 with one diagonal edge.
    \item \textbf{Squares with parallel edges}: A cycle of length 4 with either one or two opposite edges converted to parallel edges.
    \item \textbf{$D^{+1}$}: A diamond graph with one (non-diagonal) edge converted to a parallel edge (4 vertices, 6 edges). 
    \item \textbf{$K_4$ subdivide one edge}: A 5-vertex, 7-edge graph obtained by inserting a vertex of degree 2 to one of the (non-diagonal) edges of $K_4$.
\end{enumerate}

\paragraph{Solving the ILP.} For a graph $G$, we introduce an integer variable $g(b)$ for each partial assignment $b \in \{0,1,*\}^{m}$ and a variable $Q$ for the goal value, following the formulation in Bach et al. \cite{bach2018submodular}. The ILP includes the base constraint $g(*,\ldots,*) = 0$, certificate constraints $g(b) = Q$ for every $b$ containing an MST-ID certificate, and $g(b) \leq Q-1$ for every non-certificate $b$. Monotonicity is enforced via $g(b) \le g(b_{e \leftarrow v})$ for every unqueried edge $e$ and outcome $v \in \{0, 1\}$. Submodularity is enforced via $g(b_{e_1 \leftarrow v_1, e_2 \leftarrow v_2}) - g(b_{e_2 \leftarrow v_2}) \le g(b_{e_1 \leftarrow v_1}) - g(b)$ for all distinct unqueried edges $e_1, e_2$ and outcomes $v_1, v_2 \in \{0, 1\}$. We solve the resulting ILP using Gurobi, first solving the LP relaxation to obtian a lower bound, and then solving the integer program directly to obtain the optimal goal value $Q_{\text{ID}}(G)$. This approach is used for every graph in Table \ref{tab: goal values} except $K_5$.

\paragraph{$K_5$: reducing variables via isomorphism.} 
$K_5$'s $3^{10} = 59,049$ partial assignments made solving the ILP computationally infeasible. To address this, we tried reducing the number of variables of the ILP using graph automorphisms. Two partial assignments are declared isomorphic if there exists a vertex permutation $\sigma$ of $K_5$ that maps the value of each edge $\{u,v\}$ in $b$ to the value of $\{\sigma(u),\sigma(v)\}$ in $b'$. This means that any two isomorphic partial assignments have identical query states up to some relabeling, which allows for any optimal goal function to assign them the same goal value without loss of generality.

For $K_n$, we have that every vertex permutation is an automorphism. This means that $n! = 5! = 120$ automorphisms are possible for $K_5$. For each partial assignment $b$, we try all $120$ permutations, treating each resulting string as a base-3 number, and keep the smallest one as $b$'s form. Two partial assignments have the same form iff they are isomorphic under this definition. This reduces the $3^m$ partial assignments to one variable per isomorphism class in the ILP. For $K_5$, the $3^{10} = 59,049$ partial assignments reduce to 792 classes. With this reduction in place, the same ILP described above is applied directly to the 792 classes rather than the full $59,049$ partial assignments. We also tried two minor based notions of isomorphism, but both proved strictly weaker definitions than the one we used, and inflated goal value, so we do not use them. Note that enforcing equal goal values across isomorphic classes assumes that an optimal submodular goal function obeys this vertex automorphism symmetry. Under this assumption, we obtain $Q_{ID}(K_5) = 35$. Without this assumption, 35 represents a valid upper bound on $Q_{ID}(K_5)$.

\begin{table}[ht]
\centering
\begin{tabular}{lcccc}
\toprule
Graph & $n$ & $m$ & Rank$\times$corank & $Q_{\text{ID}}$ \\
\midrule
\multicolumn{5}{l}{\textit{Fan family}} \\
$K_3$ (Fan $F_2$) & 3 & 3 & 2 & 2 \\
Fan $F_3$ & 4 & 5 & 6 & 6 \\
Fan $F_4$ & 5 & 7 & 12 & 14 \\
\midrule
\multicolumn{5}{l}{\textit{Wheel family}} \\
$K_4$ (Wheel $W_3$) & 4 & 6 & 9 & 9 \\
Wheel $W_4$ & 5 & 8 & 16 & $\geq 18$\textsuperscript{*} \\
\midrule
\multicolumn{5}{l}{\textit{Simple series-parallel graphs (no parallel edges)}} \\
Cycle $C_5$ & 5 & 5 & 4 & 4 \\
Cycle $C_6$ & 6 & 6 & 5 & 5 \\
$\Theta(3,2,2)$ & 6 & 7 & 10 & 10 \\
\midrule
\multicolumn{5}{l}{\textit{Other graphs}} \\
$K_{2,3}$ & 5 & 6 & 8 & 8 \\
Square $+$ diagonal (diamond) & 4 & 5 & 6 & 6 \\
Square $+$ 1 parallel & 4 & 5 & 6 & 6 \\
Square $+$ 2 parallel pairs & 4 & 6 & 9 & 9 \\
$D^{+1}$ & 4 & 6 & 9 & 10 \\
$K_4$ subdivide one edge & 5 & 7 & 12 & 14 \\
$K_5$ & 5 & 10 & 24 & 35\textsuperscript{$\dagger$} \\
\bottomrule
\end{tabular}
\caption{Computed submodular goal values for MST-ID, compared with the rank$\times$corank value $(n-1)(m-n+1)$. \textsuperscript{*}Wheel $W_4$'s value is a lower bound and has not been proved to be optimal. \textsuperscript{$\dagger$}$K_5$'s value is computed assuming  that the optimal submodular goal function obeys the given isomorphism property.}
\label{tab: goal values}
\end{table}

For many of the graphs in the table, the $Q_{ID}$ values in Table \ref{tab: goal values} are equal to the rank$\times$corank value. $K_4$ and $F_3$ also hit the formula exactly despite us not being able to prove their behavior. Every remaining graph exceeds
the formula, smallest at the $D^{+1}$ ($+1$) and largest at $K_5$ ($+5$ on $Q_{\text{ID}}$). $K_5$ is also the only graph for which an isomorphic reduction was used.

\end{document}